\pdfoutput=1
\documentclass[11pt]{article}
\usepackage{amsmath,times,fullpage}
\usepackage{graphicx}
\usepackage{enumerate}
\usepackage{natbib}
\usepackage{url} 
\usepackage{amsthm} 

\usepackage{algorithm}
\usepackage[noend]{algorithmic}
\usepackage{authblk}
\usepackage[title]{appendix}

\newcommand{\E}{\mathrm{E}}
\newcommand{\diff}[2]{\frac{\mathit{d}^{#2}}{\mathit{dt}^{#2}} #1}
\newtheorem{theorem}{Theorem}
\newtheorem{prop}{Proposition}
\newtheorem{lem}{Lemma}

\newcommand{\dengnohat}{\ensuremath{\Psi_{r, m}(t)}}
\newcommand{\Rp}{\ensuremath{\mathit{Re}}}
\newcommand{\cv}{\ensuremath{\text{CV}}}
\newcommand{\scov}{\ensuremath{\text{SC}}}
\newcommand{\tsum}{{\textstyle\sum}}

\newcommand{\prog}[1]{\textit{#1}}

\def\spacingset#1{\renewcommand{\baselinestretch}%
  {#1}\small\normalsize} \spacingset{1}

  \title{\bf Estimating the Number of Species to Attain Sufficient Representation in a Random Sample}
  \author[1]{Chao Deng}
  \author[2]{Timothy Daley}
  \author[1]{Peter Calabrese}
  \author[1]{Jie Ren}
  \author[1,*]{Andrew D Smith
\hspace{.2cm}}
\affil[1]{Molecular and Computational Biology, University of Southern California}
\affil[*]{andrewds@usc.edu }
\affil[2]{Departments of Statistics and Bioengineering, Stanford University}

\begin{document}
\date{}
\maketitle

\bigskip

\begin{abstract}
The statistical problem of using an initial sample to estimate the
number of species in a larger sample has found important applications
in fields far removed from ecology.
Here we address the general problem of estimating the number of
species that will be represented by at least a number $r$ of observations in a
future sample.
The number $r$ indicates species with sufficient observations, which
are commonly used as a necessary condition for any robust statistical inference.
We derive a procedure to construct consistent
estimators that apply universally for a given population: once
constructed, they can be evaluated as a simple function of $r$. Our
approach is based on a relation between the number of species
represented at least $r$ times and the higher derivatives of the
expected number of species discovered per unit of time. Combining this
relation with a rational function approximation, we propose
nonparametric estimators that are accurate for both large values of
$r$ and long-range extrapolations. We further show that our estimators
retain asymptotic behaviors that are essential for applications on
large-scale datasets. We evaluate the performance of this approach by
both simulation and real data applications for inferences of the
vocabulary of Shakespeare and Dickens, the topology of a Twitter
social network, and molecular diversity in DNA sequencing data.
\end{abstract}


\noindent%
 \vfill

\newpage


\section{Introduction}\label{sec:intro}

A random sample of $N$ individuals is captured from a population after
trapping for one unit of time.
Each individual belongs to exactly one species, and the total number $L$ of
species in the population is finite but not known. Let $N_j$ be the
number of species represented by exactly $j$ individuals in this
sample, so that $N = \textstyle\sum_{j\geq 1} j N_j$.
The number of species represented $r$ or more times in the initial
sample is $S_r = \textstyle\sum_{j \geq r} N_j$.
Imagine that a second sample is obtained after trapping $t$ units of
time from the same population. The time $t > 1$ should bring to mind
a ``scaled up'' experiment. This second sample may take the form of an
expansion of the initial sample, but may also be a separate sampling
experiment as long as the second sample is representative of the
first. We are concerned with predicting the expected number
$\E[S_r(t)]$ of species represented at least $r$ times in the second
sample.

Related inference problems have been the focus of much statistical
development, with canonical applications in ecology and linguistics.
For example, \cite{zipf1935psycho,zipf1949human} was interested in the
distribution of word frequencies in random
texts. \cite{fisher1943relation} studied the relation between the
number of species and the number of individuals in a random sample;
Fisher's approach is still widely used to describe capture-recapture
experiments. When plotted as a function of $t$, the function $S_1(t)$
is called the species accumulation curve (SAC)~\citep{colwell1994estimating}.
This curve can be used to compare the diversity of populations based
on samples of differing sizes \citep{colwell2004interpolating}. More
importantly, SAC can predict the number of new
species expected in future samples. A typical question might be: given
capture profiles in a previous sample, if another sample is conducted
from the same population, how many new species would one expect to
observe in the second sample? Accurate predictions of SAC can help
scientists evaluate the future sample and allocate resources more
appropriately.

The quantity $S_1(t)$ may not be of sufficient utility when the
questions of interest involve ``common species''
\citep{preston1948commonness,pearman2007common}. In such cases the
parameter $r > 1$ in $S_r(t)$ can be naturally applied to distinguish
commonness from rarity.
In evaluating Twitter data, \cite{FM2317} focused on users with at
least $r=2$ posts, who were considered ``active'' users.
\cite{tarazona2011differential} were interested in genes represented
by more than $r=5$ sequenced reads. \cite{ng2010exome} filtered out
single nucleotide polymorphisms (SNPs) covered by fewer than $r=8$
sequenced reads. And Google Scholar uses the number of publications
cited at least $r=10$ times by others (the ``i10-index'') to measure
scholarly influence. In each of these cases a fixed $r > 1$ was used
to define those ``species'' of interest, having sufficient
multiplicity of representation in the sample.
To distinguish $S_r(t)$ from $S_1(t)$, we call $S_r(t)$ a $r$-species
accumulation curve ($r$-SAC). For the sake of convenience, we use
the terms ``SAC'' and ``$r$-SAC'' to refer to their expectations $\E[S_1(t)]$ and $\E[S_r(t)]$,
unless we explicitly say otherwise.

In this article we model frequencies of species in a sample using a
mixture of Poisson distributions
\citep{major1920an,efron1976estimating}.
In particular, individuals
representing species $i$ are assumed to be sampled according to a Poisson
process with rate $\lambda_i$ per unit of time.
The $\lambda_i$ for $i = 1, 2, \ldots, L$, can be considered as $L$
independent observations from a latent probability distribution
$G(\lambda)$.  This latent distribution describes varieties of
relative abundances among species in the population.
As a notable early example, \cite{fisher1943relation} assumed that relative
species abundance followed a gamma distribution. Although other
parametric distributions have been investigated
\citep{bhattacharya1966confluent,bulmer1974fitting,sichel1975distribution,burrell1993yes},
there are problems with using parametric distributions in practice.
There may be little information to indicate the appropriate form a
priori.
In some cases, no simple parametric form is suitable to explain
the data. In other cases, distinct parametric forms may appear to fit
the observed data well, but exhibit very different
extrapolation behaviors \citep{engen1978stochastic}.

\cite{good1956number} established a nonparametric empirical Bayes
framework that served as the foundation for much subsequent
nonparametric methodology
\citep{efron1976estimating,boneh1998estimating,chao2004nonparametric,daley2013predicting}.
\cite{good1956number} derived an estimator for the expected value of
$S_1(t)$ while avoiding direct inference of $G(\lambda)$. This
estimator takes the form of an alternating power series with
coefficients based on the count frequencies $N_j$ from the initial
sample. However, the Good-Toulmin power series usually diverges in
practice for $t>2$ \citep{good1956number}, and is consequently of
little use in modern large-scale applications.
\cite{daley2013predicting} proposed a solution to
 the divergence problem by applying
rational function approximation (RFA) to the Good-Toulmin power
series.
Development of this approach was motivated by applications
associated with DNA sequencing libraries
\citep{daley2013predicting,daley2014modeling,deng2015}, where a
``small'' sample size can be many orders of magnitude larger than
traditional ecological applications. However, the approach of
\cite{daley2013predicting} does not directly extend to $r > 1$
\citep{tim-thesis}.
Extrapolating the $r$-SAC based on an initial sample
seems more difficult when $r > 1$. In the example of
Figure~\ref{fig:curve_shape}a, the SAC appears flat after 10 units of
time, suggesting that the sample is saturated. However, for
$r = 16$, barely any species are represented at least $r$ times after
10 units of time -- leading to a very different flat curve.
Visually inspecting the shape of the $16$-SAC before 10 units
(Figure~\ref{fig:curve_shape}a) seems to provide very little
information about the shape after 20 units
(Figure~\ref{fig:curve_shape}b).

\begin{figure}[t!]
\centerline{\includegraphics{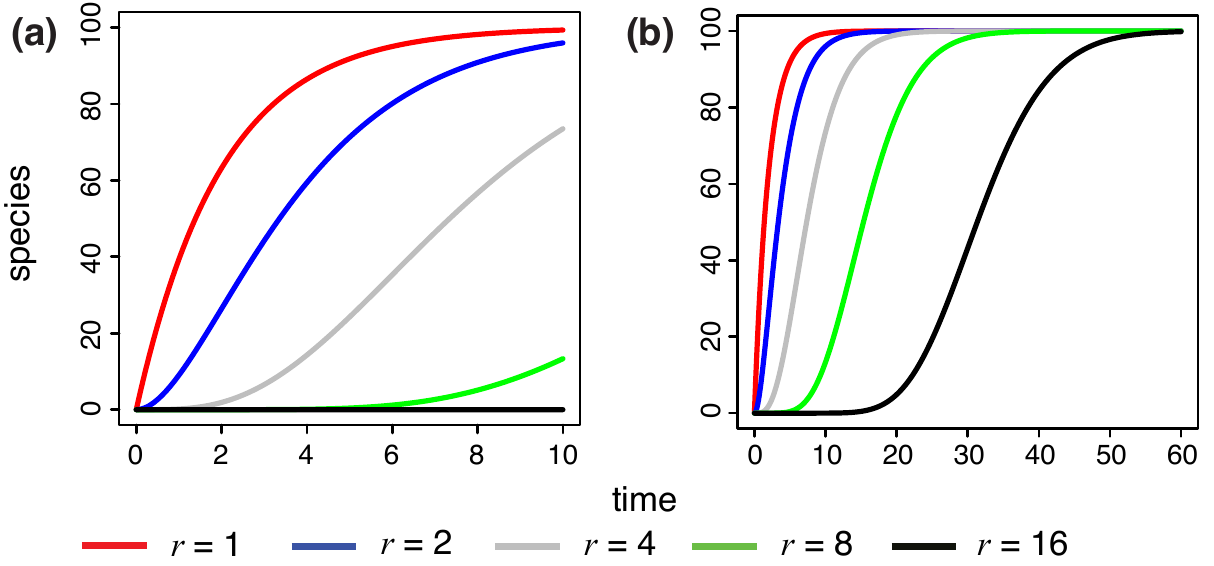}}
\caption{
Species represented by at least $r$ individuals as a
function of time $t$. 
Curves were generated from a flat model:
$\lambda_1=\lambda_2=\cdots =\lambda_{100} = 0.5$ and
population size $L = 100$. One unit of time
expects to trap 50 individuals.
The time $t$ is up to: (a) $t=10$ and
(b) $t=60$.} \label{fig:curve_shape}
\end{figure}

We describe a new approach to estimate the expected number
$\E[S_r(t)]$ of species represented at least $r$ times after trapping
for $t$ units of time, based on an initial sample from the same
population. We first derive a relation between the values we seek to
estimate and the higher-order derivatives of the average discovery
rate, defined as $\E[S_1(t)] / t$.
Then we utilize this relation to construct a universal estimator that
can apply for every value of $r$.
We show that this estimator converges in both $r$ and $t$, and is
strongly consistent as the expected size of the initial sample goes to
infinity.
Extensive simulation studies suggest that our proposed estimator
performs very well for heterogeneous populations. Applications to real
data from linguistics, social networks and DNA sequencing data confirm
the accuracy of our proposed estimator and demonstrate the value of
this new approach.

\section{Relating accumulation curves of first and higher orders}
\label{sec:s1sr}

Let $N_j$ denote the number of species
captured exactly $j$ times in an initial sample after trapping for
one unit of time, $j = 1,2,\ldots$
Clearly $N_0$ is not observable. Let $N_j(t)$ be the random
variable whose value is the number of species represented exactly $j$
times after trapping for $t$ units of time. The number $S_r(t)$ of species
represented at least $r$ times as a function of $t$ can be written as
\begin{equation}\label{eqn:defn}
S_r(t) = \sum_{j=r}^{\infty} N_j(t) = S_1(t) - \sum_{j=1}^{r-1} N_j(t).
\end{equation}
We aim to estimate the expectation of $S_r(t)$, using information
from the $N_j$.

From our Poisson mixture assumption, the expected
number of species after trapping for $t$ units of time can be expressed
\[
\E[S_1(t)] = L\int (1 - \exp(-\lambda t))dG(\lambda).
\]
Taking the $j^{\text{th}}$ derivative of $\E[S_1(t)]$, we have
\[
\diff{\E[S_1(t)]}{j} = (-1)^{j-1}L\int\lambda^j\exp(-\lambda t)dG(\lambda).
\]
Note that the expected value of $N_j(t)$ is
\[
\E[N_j(t)] ~=~ L\int \frac{(\lambda t)^j\exp(-\lambda t)}{j!}dG(\lambda) = \frac{t^j}{j!}L\int \lambda^j\exp(-\lambda t)dG(\lambda).
\]
By comparing the above expression with the $j^{\text{th}}$ derivative
of $\E[S_1(t)]$, we obtain
\begin{equation}\label{eqn:kalinin}
\E[N_j(t)] = \frac{(-1)^{j-1}t^j}{j!}~\diff{\E[S_1(t)]}{j},
\end{equation}
which has been noted previously \citep{kalinin1965functionals}.
Taking the expectation on both sides of equation \eqref{eqn:defn}, we
have
\[
\E[S_r(t)] = \E[S_1(t)] - \sum_{j=1}^{r-1} \E[N_j(t)].
\]
By replacing the $\E[N_j(t)]$ in the above equation with the
$j^{\text{th}}$ derivative of $\E[S_1(t)]$ from equation
\eqref{eqn:kalinin} we obtain a relation between $\E[S_1(t)]$ and
$\E[S_r(t)]$. This is the foundation of our estimator, and a proof can
be found in supplementary materials (Section~S1.1).

\begin{theorem}\label{thm:formula}
For any positive integer $r$,
\begin{equation}\label{eqn:main}
\E[S_r(t)] = \frac{(-1)^{r-1} t^r}{(r-1)!}
\diff{\left(\frac{\E[S_1(t)]}{t}\right)}{r-1}.
\end{equation}
\end{theorem}

Thus we have established a direct relation between the SAC and the
$r$-SAC. The quantity $\E[S_1(t)]/t$ in equation~\eqref{eqn:main}
contains information sufficient for determining $\E[S_r(t)]$, and
allows us to derive a formula for $\E[S_r(t)]$ if we are given a
smooth expression for $\E[S_1(t)] / t$.  We call the ratio $\E[S_1(t)]
/ t$ the average discovery rate, as it reflects the average rate at
which new species are discovered per unit of time.
One clear application of Theorem~\ref{thm:formula} is to generalize
existing nonparametric estimators for the SACs and obtain estimators
for the $r$-SACs. We will first demonstrate Theorem~\ref{thm:formula}
by applying it on simple parametric forms.
In the homogenous model all $\lambda_i$ are equal with $\lambda_i =
\lambda$, for $i = 1, 2, \ldots, L$, so
\[
\E[S_1(t)] = L (1 - \exp(-\lambda t)).
\]
After introducing the above expression into equation~\eqref{eqn:main},
repeatedly differentiating the quotient reveals a familiar sum:
\[
\E[S_r(t)] = L \left(1 - \sum_{i=0}^{r-1} \frac{\lambda^i\exp(-\lambda)}{i!}\right),
\]
In the negative binomial population model $\lambda_i \sim
\mathrm{Gamma}(\alpha, \beta)$ with $\alpha$ and $\beta$ positive,
\[
\E[S_1(t)] = L \left(1 - (1 + \beta t)^{-\alpha}\right).
\]
Applying Theorem~\ref{thm:formula} and the general Leibniz rule reveals the negative binomial coefficients:
\begin{align*}
\E[S_r(t)] = L\left(1 - \sum_{i=0}^{r-1}\frac{\Gamma(i + \alpha)}{\Gamma(i + 1)\Gamma(\alpha)}\left(\frac{\beta t}{1+ \beta t}\right)^i\left(\frac{1}{ 1 + \beta t}\right)^{\alpha}\right).
\end{align*}

\section{A new nonparametric estimator}
\label{sec:new_estimator}

Here we leverage the technique of Pad\'{e} approximants to build a
nonparametric estimator for the $r$-SAC. A Pad\'{e} approximant is a
rational function with a Taylor expansion that agrees with the
power series of the function it approximates up to a
specified degree
\citep{baker1996pade}. In this sense, Pad\'{e} approximants are
rational functions that optimally approximate a power series. This
method was successfully applied to construct the estimator of the SAC,
using Pad\'{e} approximants to the Good-Toulmin power series
\citep{deng2015}. Pad\'{e} approximants are effective
because they converge in practice when the Good-Toulmin power series
does not, yet within the applicable range of Good-Toulmin power series
($t < 2$), the two functions remain close. We apply the same strategy
beginning with the average discovery rate. This leads to an expression
that simplifies the formula of Theorem~\ref{thm:formula}, yielding a
new and practical nonparametric estimator for the $r$-SAC.

Our first step is to obtain a power series representation for the
average discovery rate $\E[S_1(t)]/t$ in terms of $S_i$. A proof of the following result
can be found in the supplement (Section~S1.2).

\begin{lem}\label{lem:discovery_rate_ps}
If $0 < t < 2$, then
\begin{equation}\label{eqn:expdisc}
\frac{\E[S_1(t)]}{t} = \sum_{i=0}^\infty (-1)^{i}(t - 1)^{i}\E[S_{i+1}].
\end{equation}
\end{lem}

Replacing expectations with the corresponding observations, we obtain
an unbiased power series estimator of the average discovery rate:
\begin{equation}\label{eqn:discRateEst}
\phi(t) = \sum_{i=0}^\infty (-1)^{i}(t - 1)^{i}S_{i+1}.
\end{equation}
This power series estimator $\phi(t)$ serves as a bridge between the
observed data $S_i$ and the Pad\'{e} approximant for $\E[S_1(t)] / t$,
which cannot be obtained directly.  The Pad\'{e} approximant for
$\E[S_1(t)] / t$ is defined by its behavior around $t = 1$, which is
the region where $\E[S_1(t)] / t$ is close to $\phi(t)$. Note that in
principle we could directly substitute the estimated power series
$\phi(t)$ for the average discovery rate to obtain an unbiased
power-series estimator for $\E[S_r(t)]$. Unfortunately, this estimator
practically diverges for $t > 2$, due to the small radius of
convergence of the power series and the use of the truncated power
series to approximate it (see discussion in supplemental
Section~S4).

Although Pad\'{e} approximants to a given function can have any
combination of degrees for the numerator and denominator polynomials,
we consider only the subset for which the difference in degree of
the numerator and denominator is 1. This choice permits these rational
functions to mimic the long-term behavior of the average discovery
rate, which should approach $L/t$ for large $t$.

Let $P_{m-1}(t)/Q_m(t)$ denote the Pad\'{e} approximant to power
series $\phi(t)$ with numerator degree $m - 1$ and denominator degree
$m$. According to the formal determinant
representation \citep{baker1996pade},
\begin{multline}\label{eqn:matrixform}
\footnotesize
\frac{P_{m-1}(t)}{Q_{m}(t)} = \frac{a_0 + a_1(t - 1) + \cdots + a_{m-1}(t-1)^{m-1}}{b_0 + b_1(t-1) + \cdots + b_m(t-1)^m} = \\[0.5em]
{\footnotesize
\frac{
\begin{vmatrix}
    (-1)^0 S_1 & (-1)^1 S_2 & \dots &  (-1)^{m-1} S_{m} & (-1)^m S_{m+1}\\
    (-1)^1 S_2 & (-1)^2 S_3 & \dots & (-1)^{m} S_{m+1} & (-1)^{m+1} S_{m+2} \\
    \vdots & \vdots & \ddots & \vdots & \vdots \\
    (-1)^{m-1} S_{m} & (-1)^{m} S_{m+1} & \dots &  (-1)^{2m-2} S_{2m -1} & (-1)^{2m-1} S_{2m} \\
    0 & \!\!\!\!(-1)^0 S_1(t-1)^{m-1} \!\!\!& \dots & \!\!\!\sum_{i=0}^{m-2} (-1)^i S_{i+1} (t-1)^{i+1} & \sum_{i=0}^{m-1} (-1)^i S_{i+1} (t-1)^i
\end{vmatrix}}{
\begin{vmatrix}
    (-1)^0 S_1 & (-1)^1 S_2 & \dots &  (-1)^{m-1} S_{m} & (-1)^m S_{m+1}\\
    (-1)^1 S_2 & (-1)^2 S_3 & \dots & (-1)^{m} S_{m+1} & (-1)^{m+1} S_{m+2} \\
    \vdots & \vdots & \ddots & \vdots &  \vdots \\
    (-1)^{m-1} S_{m} & (-1)^{m} S_{m+1} & \dots &  (-1)^{2m-2} S_{2m -1} & (-1)^{2m-1} S_{2m} \\
    (t - 1)^m & (t-1)^{m-1} & \dots & (t - 1) & 1
\end{vmatrix}}.
}
\end{multline}
The above representation allows us to reason algebraically about the
existence of the desired Pad\'{e} approximant to $\phi(t)$ for a given initial
sample.
Define the Hankel determinants
\begin{equation}\label{eqn:delta}
\Delta_{i, j} =
\begin{vmatrix}
  S_{i - j + 2} & S_{i - j + 3} & \dots  & S_{i+1} \\
  S_{i - j + 3} & S_{i - j + 4} & \dots  & S_{i+2} \\
  \vdots & \vdots &  \ddots & \vdots \\
  S_{i+1} & S_{i+2} & \dots  & S_{i+j}
\end{vmatrix}_{j\times j},
\end{equation}
with $S_k = 0$ for $k<1$.  A proof of the next lemma is given in
the supplement (Section~S1.2).

\begin{lem}\label{lem:existpade}
If the determinants $\Delta_{m-1, m}$ and $\Delta_{m, m}$ are nonzero,
there exist real numbers $a_i$ and $b_j$ for $i = 0, \ldots, m - 1$
and $j = 1, 2, \ldots, m$, with $b_m \neq 0$, such that the rational function
\[
\frac{P_{m-1}(t)}{Q_m(t)} = \frac{a_0 + a_1(t - 1) + \cdots + a_{m-1}(t-1)^{m-1}}{1 + b_1(t-1) + \cdots + b_m(t-1)^m}
\]
satisfies
\begin{equation}\label{eqn:padecondition}
\phi(t) - \frac{P_{m-1}(t)}{Q_m(t)} = O\left((t - 1)^{2m}\right),
\end{equation}
and all $a_i$ and $b_j$ are uniquely determined by $S_1,S_2,\ldots,S_{2m}$.
\end{lem}

In what follows we assume that denominators of all rational functions
of interest have simple roots. In practice we do not encounter
$Q_m(t)$ with repeated roots, and in the supplement we show how this
assumption can be removed (Section~S1.2).

\begin{theorem}\label{thm:exist}
Let $m$ be a positive integer. If both determinants $\Delta_{m-1,m}$
and $\Delta_{m,m}$ are nonzero, then there exist complex numbers
$c_i$ and $x_i$, uniquely determined by $S_1$, $\ldots$, $S_{2m}$, such that for all $1 \leq r \leq 2m$,
\begin{equation}\label{eqn:estimator}
\dengnohat{} = \sum_{i=1}^m c_i\left(\frac{t}{t - x_i} \right)^r
\end{equation}
satisfies $\Psi_{r, m}(1) = S_r$.
\end{theorem}

\begin{proof}
  The assumptions that $\Delta_{m-1, m} \neq 0$ and $\Delta_{m, m} \neq 0$ imply that
  the Pad\'{e} approximant $P_{m-1}(t) / Q_m(t)$ exists in
  correspondence with $\phi(t)$. Substituting $P_{m-1}(t)/Q_m(t)$ in
  place of the average discovery rate in equation~\eqref{eqn:main}, we
  define
  \begin{equation}\label{eqn:rfaSub}
    \dengnohat{} =
    \frac{(-1)^{r-1} t^r}{(r-1)!}
    \diff{\left(\frac{P_{m-1}(t)}{Q_m(t)}\right)}{r-1}.
  \end{equation}
  By the definition of the Pad\'{e} approximant, we have
  \begin{equation}\label{rfa:error}
    \phi(t) - \frac{P_{m-1}(t)}{Q_m(t)} = O\left((t - 1)^{2m}\right).
  \end{equation}
  Taking derivatives of $\phi(t)$ at $t = 1$, for $j = 0, 1, \ldots, 2m - 1$,
  \[
  \diff{\left(\frac{P_{m-1}(1)}{Q_m(1)}\right)}{j} =  \diff{\phi(1)}{j}= (-1)^j j! S_{j+1}.
  \]
  Therefore, for any $r=1,2,\ldots,2m$,
  \begin{equation}
    \Psi_{r, m}(1) = \frac{(-1)^{r-1}}{(r-1)!}
    \diff{\left(\frac{P_{m-1}(1)}{Q_m(1)}\right)}{r-1} = \frac{(-1)^{r-1}}{(r-1)!} \Big((-1)^{r-1} (r-1)!S_r\Big) = S_r.
  \end{equation}

  Now we show that \dengnohat{} defined in \eqref{eqn:rfaSub} can be
  expressed in the desired form \eqref{eqn:estimator}. Let
  $x_1,\ldots,x_m$ be the distinct roots of $Q_m(t)=0$. We can write
  $P_{m-1}(t) / Q_m(t)$ as
  \begin{equation}\label{eqn:partialfraction}
    \frac{P_{m-1}(t)}{Q_m(t)} = \sum_{i=1}^m \frac{c_i}{t - x_i},
  \end{equation}
  where $c_i$ are coefficients of the partial fraction
  decomposition. The required derivatives take a convenient form:
  \begin{equation*}
    \diff{\left(\frac{c_i}{t - x_i}\right)}{r-1}
    = ~
    (-1)^{r-1} (r-1)! \left(\frac{c_i}{(t - x_i)^r}\right).
  \end{equation*}
  By substituting these derivatives into \eqref{eqn:rfaSub} we arrive
  at
  \begin{equation}\label{eqn:rfaEst}
    \dengnohat{} = \frac{(-1)^{r-1} t^r}{(r-1)!}
    \diff{\left( \sum_{i=1}^m \frac{c_i}{t - x_i} \right)}{r-1} = \sum_{i=1}^m c_i\left( \frac{t}{t - x_i} \right)^r.
  \end{equation}

  Finally, the uniqueness of the coefficients $c_i$ and the roots $x_i$
  follows from the uniqueness of the Pad\'{e} approximant $P_{m-1}(t) / Q_m(t)$,
  which is a function of $S_i$, $i=1, \ldots, 2m$.
\end{proof}

The function \dengnohat{} in Theorem~\ref{thm:exist} is a
nonparametric estimator for the $r$-SAC. Of note, the coefficients
$c_i$ and poles $x_i$ are independent of $r$: once determined, they
can be used to directly evaluate \dengnohat{} for any $r$.  The
estimator \dengnohat{} has some favorable properties, summarized in
the following proposition, with proofs given in
Section~S1.3.

\begin{prop}\label{prop:estimatorproperty}
(i) The estimator \dengnohat{} is unbiased for $\E[S_r(t)]$ at $t = 1$ for $r \leq 2m$.\\
(ii) The estimator \dengnohat{} converges as $t$ approaches infinity. In particular,
\[
\lim_{t\to\infty} \dengnohat = \frac{\Delta_{m-1, m+1}}{\Delta_{m, m}}.
\]
(iii) The estimator \dengnohat{} is strongly consistent as the
initial sample size $N$ goes to infinity.
\end{prop}

\noindent Remark. Both determinants $\Delta_{m-1,m}$ and $\Delta_{m,m}$ become 0 when
$S_j = L$ for $j \leq 2m$ and $m > 1$, so the determinant
representation of the Pad\'{e} approximant \eqref{eqn:matrixform} is
ill-defined in such cases. However, the Pad\'{e} approximant itself
remains valid and reduces to $L / t$ for $t > 0$
(see Section~S1.3).


\section{An algorithm for estimator construction}
\label{sec:algorithm}

\subsection{Conditions for well-behaved rational functions}

The choice of $m$ controls the degree of both the numerator and the
denominator in the Pad\'{e} approximant, and determines the amount of
information from the initial sample that is used by \dengnohat{}. In
principle $m$ should be selected sufficiently large so that the
estimator \dengnohat{} can explain the complexity of the latent
distribution $G(\lambda)$.  However, a larger value of $m$ leads to
more poles in the estimator \dengnohat{} and makes instability more
likely.  In practice, the stability of the estimators depends on the
locations of poles. For example, if any pole $x_i$ resides on the
positive real axis, then \dengnohat{} is unbounded in the neighborhood
of $x_i$ and becomes ill-defined at $t = x_i$. Here we give a
sufficient condition to stabilize the estimator so that it is
well-defined and bounded for $t \geq 0$ and $r \geq 1$. Moreover, this
condition ensures that as $r$ approaches infinity, the estimator
\dengnohat{} approaches zero for fixed $t$. Note $\Rp(x)$ is
the real part of $x$. A proof of the next proposition is given in the supplement
(Section~S1.4).
\begin{prop}\label{prop:pacman}
If $\Rp(x_i) < 0$ for $1 \leq i \leq m$, then \dengnohat{} is bounded for any $t \geq 0$ and
$r \geq 1$.
Further, $\dengnohat{}\rightarrow 0$
as $r\rightarrow \infty$ for any $0\leq t < \infty$.
\end{prop}

\noindent Remark. It is not unusual to constrain roots in such a way to ensure
stability. For example, the Hurwitz polynomials, which has all zeros
located in the left half-plane of the complex plane, are used as a
defining criterion for a system of differential equations to have
stable solutions.

\subsection{The construction algorithm}
\label{sec:algorithm}

\begin{algorithm}[h!]
\vspace{0.1pc}
 \begin{algorithmic}[1]
   \STATE Compute sums $S_i = \sum_{j \geq i} N_j$, for $i = 1,
   \ldots, 2m_\mathrm{max}$. These coefficients define
   $\phi(t)$.
   \STATE Compute the coefficients of the degree $2m_\mathrm{max}$ continued fraction approximation to $\phi(t)$ by applying the quotient-difference algorithm.
   \FOR {$m \gets m_\mathrm{max}$ \TO $1$}
   \STATE Obtain the Pad\'{e} approximant $P_{m-1}(t) / Q_m(t)$ by evaluating the
   $2m$-th convergent (truncation) of the continued fraction.
   \STATE Obtain the roots $x_i$, for $i = 1,\ldots,m$, of the denominator $Q_{m}(t)$.
   \STATE Calculate coefficients $c_i$ by partial fraction decomposition
   of $P_{m-1}(t) / Q_m(t)$.
 \IF {$\Rp(x_i) < 0$ for all $x_i$ and $\Psi_{1, m}(t)$ is increasing}
 \RETURN coefficients $(c_1, \ldots, c_m)$ and roots $(x_1, \ldots, x_m)$.
 \ENDIF
 \ENDFOR
 \end{algorithmic}
  \caption{Given a set of observed counts $\{N_j\}$, with $N_1, N_2 >0$, and a maximal
   value of $m_\mathrm{max}$, produce the stable and increasing estimator
   $\Psi_{r, m}(t)$ for maximal $m \leq m_\mathrm{max}$.}\label{alg:construction}
\end{algorithm}

Algorithm~\ref{alg:construction} provides a complete procedure for
constructing our estimator beginning with the observed counts $N_j$,
and satisfying the conditions outlined above. This procedure requires
specifying a maximum value of $m$, but also leaves room for using more
effective numerical procedures at each step.
Details about these procedures can be found in the
supplementary materials (Section~S3).

To see that Algorithm~\ref{alg:construction} terminates successfully,
note that when $m = 1$, 
\begin{equation}\label{estimator:m=1}
  \Psi_{r, 1}(t) = \frac{S_1^2}{S_2}\left( \frac{t}{t + (S_1 - S_2) / S_2} \right)^r.
\end{equation}
So if there exist at least one species represented once and one species
represented more than once in the initial sample, then we observe $S_1
- S_2 > 0$ and $S_2 > 0$.  This ensures $\Psi_{r,1}(t)$ satisfies
$\Rp(x_i) < 0$ and $\Psi_{r,1}(t)$ is increasing for every $r \geq
1$.

\subsection{Variance and confidence interval}

Deriving a closed-form expression for the variance of the estimator
\dengnohat{} is challenging. On one hand, when $m \geq 5$ we have no
general algebraic solution to the polynomial equations that identify
$x_i$ in \dengnohat{}, so a closed-form may not exist. On the other
hand, even for $m = 1$ the variance of $\Psi_{r,1}(t)$ involves a
nonlinear combination of random variables $S_1$ and $S_2$
(equation~\eqref{estimator:m=1}).

In practice we approximate the variance of our estimates by bootstrap
\citep{efron1994introduction}. Each bootstrap sample is a vector of
counts
\[
(N_1^*, N_2^*,\ldots, N_{j_\mathrm{max}}^*)
\]
that satisfies $\tsum_{i=1}^{j_\mathrm{max}} N_i^* = S_1,$ where
$j_\mathrm{max}$ is the largest observed frequency for a species in
the initial sample and $S_1$ is the number of species observed in the
initial sample. The $(N_1^*, N_2^*, \ldots, N_{j_\mathrm{max}}^*)$ is
sampled from a multinomial distribution with probability in proportion
to $(N_1, N_2, \ldots, N_{j_\mathrm{max}})$.  For each bootstrap, we
construct an estimator $\Psi_{r,m}^*(t)$ for the $r$-SAC.  All
estimators $\Psi_{r,m}^*(t)$ are then used to calculate the variance
of the estimator \dengnohat{}. Estimating confidence intervals as
percentiles of the bootstrap distribution requires too many samples
(e.g. \citet[Chapter~13]{efron1994introduction} suggest 1000) for
large-scale applications.  Instead we adopt the lognormal approach,
where the mean and variance can be accurately estimated using far
fewer bootstrap samples. Use of the lognormal is justified by an
observed natural skew for quartiles of estimates in our simulation
results (Figure~\ref{fig:ds_se_compare}a).

\section{Simulation studies}
\label{sec:sim}

We carried out a simulation study to assess the performance of the
estimator \dengnohat{}. The simulation scheme is partly inspired by
\cite{chao2004nonparametric} but involves populations and samples of
larger scale. Following our statistical assumptions, the number of
individuals for species $i$ in the initial sample follows a Poisson
distribution with the rate $\lambda_i$, for $i = 1, 2, \ldots, L$. The
rates $\lambda_i$ are generated from distributions we have chosen to
model populations with different degrees, types of heterogeneity and sample coverage.
We measure the {\it degree} of heterogeneity in a population by the
coefficient of variation (\cv{}) for
$\lambda_i$:
\begin{equation}\label{eqn:cv}
\bar{\lambda}^{-1}\Big((L-1)^{-1}\tsum_{i=1}^L(\lambda_i - \bar{\lambda})^2\Big)^{1/2},
~~\mbox{ where }~~
\bar{\lambda} = \tsum_{i=1}^L \lambda_i\slash L.
\end{equation}
The coefficient of variation quantifies difference in relative
abundances among species and is independent of sample sizes. For the
{\it type} of heterogeneity, we focus on the shapes of distributions,
for example distinguishing those with exponentially decreasing tail
versus heavy-tailed distributions. Sample coverage (\scov{}) is
defined as the total proportion of species in the population that are
covered in the sample. Sample coverage is one indicator for how well a
sample can represent the corresponding population: relatively little
can be inferred about those species not observed.

We selected six models for our simulations. The first is a homogenous
model, the Poisson distribution (P), included as a basis for
comparison with the other models. Intuitively, the homogeneous model
is the simplest one among all models. However, for a given sample
size, samples from the homogeneous population have the least coverage
among any type of population if the sample size is not too large
(See details and the proof in the supplementary materials).  The
second and third models are negative binomial (NB1 and NB2), where the
$\lambda_i$ follow gamma distributions. The NB models are widely used
to describe overdispersed counts data \citep{hilbe2011negative}.
The fourth model is a lognormal (LN) model \citep{bulmer1974fitting},
which has been applied in ecology \citep{preston1948commonness}.
Models 5 and 6 are a Zipf distribution (Z; \citealp{zipf1935psycho})
and a Zipf-Mandelbrot distribution (ZM;
\citealp{mandelbrot1977fractals}), respectively,
which are known as power law.
Models 4--6 represent so called heavy-tailed populations
\citep{newman2005powerlaw}.  Table~\ref{tab:simuparameter}
summarizes these parameter settings.

In our simulations we fixed the total number $L$ of species at 1
million (M) to represent large-scale applications.
For the results
below, the expected size of initial samples was also set to 1M
individuals. For each model, the values of parameters in each model
were determined in a way such that
$\tsum_{i=1}^L \lambda_i = L = 1 \text{M}.$
Our simulations covered $(t, r)$ representing the region
$[1,100]\times\{1,\ldots,100\}$, which more than covers the $(t, r)$
we have seen in practical applications. We measure
performance of estimators using relative error. For fixed $r$, relative
error is calculated as the $L^2$-distance between the expected
$\E[S_r(t)]$ and the estimate, divided by the $L^2$-norm of
$\E[S_r(t)]$, 
evaluated at
$t = 1, 2, \ldots, 100$.
The errors we report are means of relative error over the curves
for $r = 1, 2, \ldots, 100$.

We compared the estimator \dengnohat{} with several other estimators.
The zero-truncated Poisson (ZTP; \citealp{cohen1960estimating}) and
zero-truncated negative binomial (ZTNB; \citealp{sampford1955truncated})
are obvious and expected to perform well when the underlying
statistical assumptions of the estimator matches the model of the
simulation. The logseries (LS) approach, popularized in ecology, was
introduced as a special case of the ZTNB method when the shape
parameter in the negative binomial distribution was close to $0$
\citep{fisher1943relation}. To our knowledge, there is no
nonparametric estimator designed for $\E[S_r(t)]$ when $r > 1$. To
evaluate other plausible approaches, we made use of two nonparametric
estimators for SACs, specifically those due to
\cite{boneh1998estimating} and \cite{chao2004nonparametric}, which we
refer to as BBC and CS, respectively.  We leveraged
equation~\eqref{eqn:main} in Theorem~\ref{thm:formula} to derive
general estimators of $\E[S_r(t)]$, for $r\geq 1$, based on these two
estimators of $\E[S_1(t)]$. These derivations can be found in the
supplementary materials (Section~S2).

\begin{table}[t!]
\caption{Models used in simulations, with corresponding parameter
  settings, \cv{} and \scov{} values. For models Z and ZM, the \cv{} values
  are calculated using equation~\eqref{eqn:cv}. For other models,
  listed values for \cv{} are expectations directly calculated based
  on the underlying distributions and
  parameters. \scov{} is based on mean of sample coverage over 1000 samples.}
  \label{tab:simuparameter}
\begin{tabular*}{\columnwidth}{@{\extracolsep{\fill}}lllcc} \\ \hline
  Model & Name & Distribution on rates & $\cv{}$ & \scov{}\\
  \hline
 P & Homogeneous & $\lambda_i \propto 1$ & 0 & .632\\
 NB1 & Negative binomial & $\lambda_i \sim$ $\Gamma$$($shape=1, scale=1$)$ & $1$ & .750\\
 NB2 & Negative binomial & $\lambda_i \sim$ $\Gamma$$($shape=0.01, scale=1$)$ & $10$ & .991\\
 LN & Poisson-lognormal &  $\log(\lambda_i) \sim$ Gaussian$(0, 1)$ & $1.31$ & .742\\
 Z & Poisson-Zipf & $\lambda_i  \propto$ $1 / (i + 100)$ & $10.79$ & .810\\
 ZM & Poisson-Zipf-Mandelbrot & $\lambda_i   \propto$ $1 / (i + 100)^{1.1}$ & $15.10$ & .849\\
  \hline
\end{tabular*}
\end{table}

\subsection{Simulation results}

As can be seen from Figure~\ref{fig:ds_se_compare}a, the estimator
\dengnohat{} performs well under models NB1, NB2, LN, Z and ZM. We
consider these to represent heterogeneous populations due to their
large \cv{} compared with the homogeneous model
(Table~\ref{tab:simuparameter}). The relative errors are $0.002$ ($\pm
0.003$) and $0.027$ ($\pm 0.011$) for NB1 and NB2. The errors for the
Z and ZM models are slightly higher: $0.057$ ($\pm 0.042$) and $0.057$
($\pm 0.040$), respectively (Table~\ref{tab:simuResult}). Both
the relative error and the standard error of \dengnohat{} are much higher
when applied to the homogenous models
(Figure~\ref{fig:ds_se_compare}a).

We compared the estimator \dengnohat{} with the five other estimators.
The estimator \dengnohat{} has the least mean relative error compared
with other approaches under the LN, Z and ZM models (Figure~\ref{fig:ds_se_compare}b), 
which are the heavy-tailed models.  The
relative errors under these three models are $0.020$, $0.057$ and
$0.057$ (Table~\ref{tab:simuResult}). In particular, under the Z
and ZM models, the second most accurate approach, our generalization
of CS estimator, has relative error $0.525$ and $0.558$,
around $10\times$ the error of \dengnohat{}. The estimator
\dengnohat{} has higher standard error compared with the other methods
(Figure~\ref{fig:ds_se_compare}b), which we attribute broadly to
its use of procedures ({\it e.g.} to fit the Pad\'{e} approximant)
that can introduce numerical error. Even considering this variation,
when \dengnohat{} is at its least accurate it remains substantially
more accurate than the other methods across models LN, Z and ZM.  As
expected, for model NB1 and NB2, the ZTNB approach is the most
accurate because it matches the precise statistical assumptions of
those simulations. Importantly, without any assumption about the
latent distribution of $\lambda_i$, the estimator \dengnohat{} also
yields excellent accuracy in these two models, with relative errors
less than $5\%$. The LS approach performs similar to the ZTNB
approach when the shape parameter in the NB model is close to zero, as
occurs for NB2 (Figure~\ref{fig:ds_se_compare}). Similarly, for
the homogeneous population model the ZTP approach is the most
accurate.

\begin{figure}[t!]
\centerline{\includegraphics{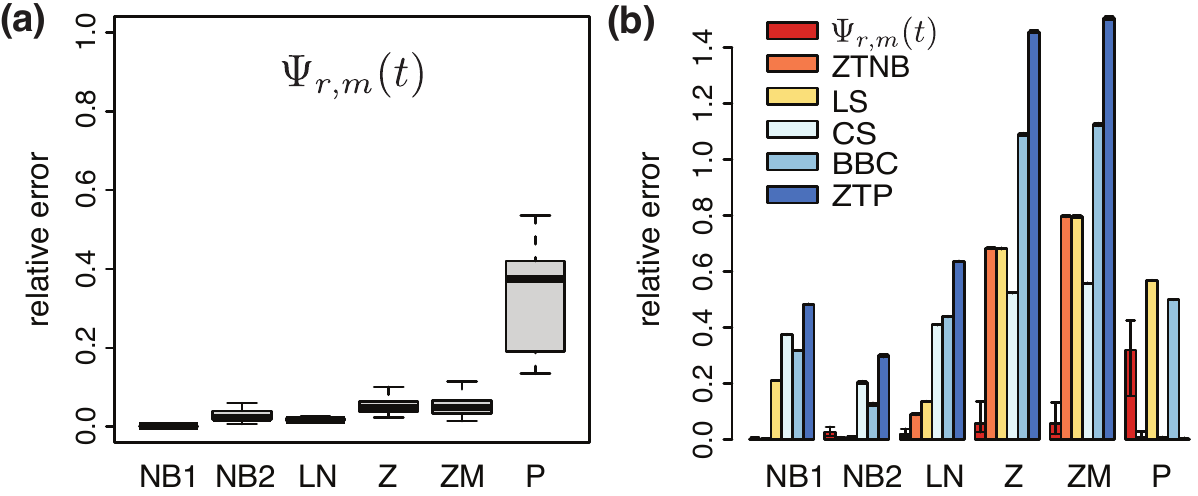}}
\caption{Relative errors in simulation studies.
 (a) relative error of the estimator \dengnohat{} for the
  six simulation models. Box plots are based on 1000 replicate
  simulations. The horizontal bar displays the median, boxes display
  quartiles and whiskers depict minima and maxima. (b) Mean relative
  error of all tested estimators for simulated datasets based on 1000 replicates
  for each model. The error bars show the 95\% confidence interval
  of relative errors.}
\label{fig:ds_se_compare}
\end{figure}

We found the estimator \dengnohat{} to be more accurate when the population
samples correspond to heavy-tailed distributions compared with other methods.
In general, these are the most challenging
scenarios for accurately predicting $\E[S_r(t)]$
(Figure~\ref{fig:ds_se_compare}b).
The NB2 and Z models have a similar degree of heterogeneity in terms of \cv{}
(Table~\ref{tab:simuparameter}),
but for all estimators except \dengnohat{}, relative error for Z is clearly
larger than the error for NB2.
This difference is associated with the change from exponentially
decreasing (NB2) compared with the power law distribution. For
\dengnohat{}, the relative error remains small in both these
scenarios. The above results correspond to an initial sample size of
$N=L$, but for initial samples of $0.5L$ to $2L$ the mean relative
error changed very little for the heterogeneous models (Figure~S1).
The error only noticeably increased when the
sample size was below $0.4L$.

Clearly our estimator has larger relative errors when the samples are
generated from a
homogeneous model compared with other models
(Figure~\ref{fig:ds_se_compare}). Our initial intuition was that the
homogeneous model should be easier to prediction because all
$\lambda_i$ are constrained by a single parameter. Our simulation
results show an interesting dichotomy in the performance of the
methods we tested. On one hand, nonparametric methods that do not
assume an underlying Poisson have higher relative error on the
homogeneous model. For example, the relative errors are 0.5 and 0.32
for BBC and our estimator (Table~\ref{tab:simuResult}). On the other
hand, relative error is 0.003 for CS, which is based on the Poisson
distribution. Parametric methods show similar trends. The ZTNB
performs well under the homogeneous model because it can easily
describe a Poisson when the shape parameter is large. Although the LS
estimator is derived from the negative binomial, it assumes the shape
parameter is close to 0, so it has difficulty describing homogeneous
data.

The sample coverage provides one perspective on why the homogeneous
model might present challenges for nonparametric approaches. In
particular, the homogeneous model has the lowest sample coverage
compared with other models having a fixed sample size (see Section~S6). 
Increasing the initial sample size can increase sample
coverage, which in turn improves the accuracy of our estimator. For
example, when we increase the size of the initial sample to 2M, the
relative error reduces to 0.123 ($\pm 0.06$).

\begin{table*}[t]
\caption{Relative error and standard error for the six simulation
  models. Numbers in parentheses are standard errors based on 1000 replicates.
} \label{tab:simuResult}
{\scriptsize
\begin{tabular*}{\columnwidth}{@{\extracolsep{\fill}}l*{6}r}
\\ \hline
& P & NB1 & NB2 & LN & Z & ZM\\ \hline
\dengnohat{} & .320 (.112) & .002 (.003)  & .027 (.011)  & .020 (.014)  & .057 (.042)  & .057 (.040) \\
ZTNB & .008 (.008) & .001 (.001)  & .004 (.002)  & .090 (.001)  & .683 (.001)  & .798 (.001) \\
LS & .569 (.000) & .211 (.000)  & .006 (.002)  & .137 (.000)  & .682 (.001)  & .797 (.001) \\
CS & .003 (.002) & .375 (.001)  & .204 (.002)  & .410 (.001)  & .525 (.001)  & .558 (.001) \\
BBC & .500 (.000) & .319 (.000)  & .126 (.003)  & .439 (.001)  & 1.090 (.002)  & 1.126 (.002) \\
ZTP & .002 (.002) & .484 (.000)  & .299 (.002)  & .637 (.001)  & 1.456 (.002)  & 1.505 (.003) \\
\hline
\end{tabular*}}
\end{table*}

\subsection{Best practice}

Based on simulations, we found that the estimator \dengnohat{} is accurate
when populations are heterogeneous
(Figure~\ref{fig:ds_se_compare}a). It suffers large relative
errors and variance when populations are close to being homogeneous, a
context where the ZTNB works well (Table~\ref{tab:simuResult}).
Our best-practice advice is to combine both our estimator \dengnohat{}
and the ZTNB. Whenever samples are generated from a heterogeneous
population, we should use the estimator \dengnohat{}; otherwise, we
switch to the ZTNB estimator to handle the homogeneous cases.
We use the coefficient of variation (\cv{}) to measure the degree of
heterogeneity in a population.  In practice, whenever the estimated
\cv{} is greater than $1$, we use our estimator \dengnohat{};
otherwise, we switch to the ZTNB estimator.  The procedure of
estimating the \cv{} and the rationality for using $\cv{} = 1$ as the
cutoff can be found in the supplementary materials
(Section~S5).
In our simulations all samples from the P model
have estimated \cv{}s less than 1 and all samples from the NB2, LN, Z and ZM model
have estimated \cv{}s greater than 1 (Figure~S2a).
Estimated \cv{}s are around 1 for the NB1 model, in which both \dengnohat{} and
the ZTNB approach give accurate estimates (Table~\ref{tab:simuResult}).



\section{Applications}
\label{sec:largescale}

We applied our estimator to data from three different domains:
linguistics, a social network, and a DNA sequencing application.  In
each case the data may be considered ``big''. We adopt a
strategy of sub-sampling from the full available data to generate a
ground truth reference for evaluation. We include the ZTNB for
comparison due to its popularity for overdispersed counts data \citep{hilbe2011negative}.
The estimated \cv{} for each dataset is in Table~S1.

\subsection{The vocabulary of Shakespeare and Dickens}

We first re-examined the Shakespearean vocabulary problem due to
\cite{efron1976estimating}. The data is 884,647 words written,
corresponding to a set of 31,534 distinct words.  There are 14,376
distinct words that appear exactly once in the collection, 4,343 that
appear exactly twice, and so on. The full word appearance frequencies
are listed in Table 3 by \cite{efron1976estimating}. Our task is
to predict the number of distinct words that would appear at least $r$
times if some additional quantity of Shakespeare's work is discovered.
For a special case $r = 1$, the problem has been discussed by
previous studies \citep{efron1976estimating}. Compared our prediction
results with previous studies, we found that the results are
surprisingly consistent.

The numbers $S_j$ of distinct words that appear in the collection at
least $j$ times, for $j = 1, 2, \ldots, 20$, are given in
Table~S2.  We applied
Algorithm~\ref{alg:construction} and obtained:
{\small
\[
\dengnohat{} = 120357.66\left(\frac{t}{t + 14.91}\right)^r + 24934.99\left(\frac{t}{t + 1.13}\right)^r + 13453.12\left(\frac{t}{t + 0.10}\right)^r.
\]}

The estimator $\Psi_{1, m}(t)$ predicts 42,993 ($\pm 586.17$)
distinct words when $t = 2$ ({\it i.e.} the unlikely event that ``the
other half'' of Shakespeare's were to be discovered). The additional
work is expected to contain 11,459 new distinct words.  The corresponding
prediction by Good and Toulmin's estimator is 11,430, and the
prediction by Fisher's negative binomial model is 11,483
\citep{efron1976estimating}. Prediction results of $\Psi_{1, m}(t)$
for $t = 4, 6, 11, 21$ are shown in Table~\ref{tab:shakespeareP}.  All
these estimates are consistent with the estimation by
\cite{efron1976estimating}.

\begin{table}[t]
\vspace{0.5pc}
\caption{Comparison of predictions by \dengnohat{} and Efron-Thisted estimator for $r = 1$.
  The expected number of distinct
  words when a total of $884647\times t$ words are discovered. The first and second
  columns are estimates and standard error by $\Psi_{1,m}(t)$. The fourth
  and fifth columns are lowerbound and upperbound estimated by
  Efron's estimator \citep[Table 5]{efron1976estimating}.}
\begin{tabular*}{\columnwidth}{@{\extracolsep{\fill}}l*{4}r} \label{tab:shakespeareP}
\\ \hline
  &  &  & \multicolumn{1}{c}{lower bound by} & \multicolumn{1}{c}{upper bound by}  \\
$t$ & \multicolumn{1}{c}{$\Psi_{1,m}(t)$} & SE & \multicolumn{1}{c}{Efron-Thisted estimator} & \multicolumn{1}{c}{Efron-Thisted estimator} \\ \hline
2 & 11,459 & 586 & 11,205 & 11,732 \\
4 & 26,494 & 1,171 & 23,828 & 29,411 \\
6 & 37,215 & 2,582 & 29,898 & 45,865 \\
11 & 55,501 & 7,675 & 34,640 & 86,600 \\
21 & 75,894 & 18,765 & 35,530 & 167,454 \\ \hline
\end{tabular*}
\end{table}

Among Shakespeare's known works, 17,158 words appear at least
twice. When $t = 2$ and $r= 2$, \dengnohat{} predicts that a total of
24,101 distinct words are expected to be observed at least twice. So
there are 6,943 new words observed at least twice when doubling
the amount of text. These new words could be either from
Shakespeare's known work that are observed exactly once, or from words
observed only in the additional work.  A total of 14,376 distinct words
appeared exactly once in Shakespeare's known work.  So at least
7,433 ($14376 - 6943$) distinct words that appear once in Shakespeare's
known work are likely to be absent from newly discovered work
of the same size.

We also applied the estimator \dengnohat{} to infer word frequencies
in a sample of Charles Dickens' work. We used data from Project
Gutenberg
as included with the R package
\prog{zipfR} (v0.6-6) \citep{zipfR2007}. This data set contains
roughly 2.8M written words,
of which just over
41k are distinct.
We sampled 300k words from the dataset as an initial sample and
applied \dengnohat{} for values of $r$ between 1 and 20.
Figure~\ref{fig:dickens} shows the estimated curves along with actual
curves from the entire Dickens data set. The estimated curves track
the true curves very closely. In contrast, ZTNB is
inaccurate for both $r = 1$ and $r > 1$. Table~S3 shows
estimated values and their standard errors (SE) for extrapolations of $5\times$ and
$9\times$, the latter is the maximum possible given the size of the
data set. Even at $r=20$ the relative error never exceeds $5\%$.
We also examined the behavior of \dengnohat{} as a function of $r$. As can
be seen from Figure~S3a, \dengnohat{} remains
accurate for large values of $r$. In comparison, ZTNB tends to
overestimate the observed values.

\begin{figure}[t]
\centerline{\includegraphics{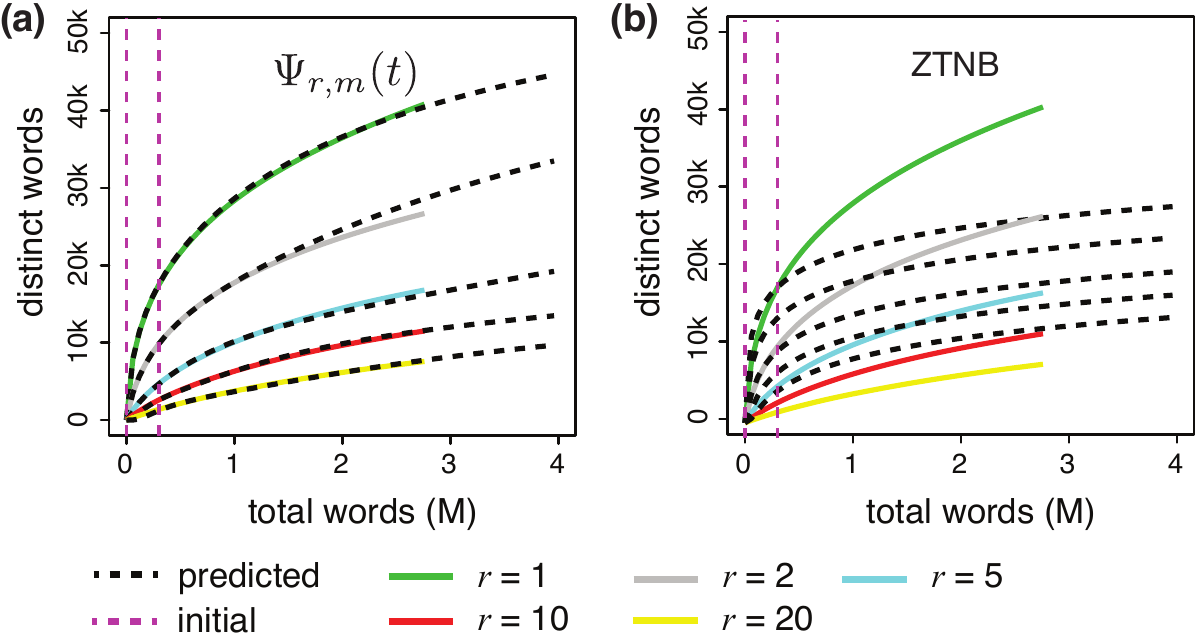}}
 \caption{distinct words represented at least $r$ times in
     samples of Dickens' work as a function of words written.
Estimates are based on an initial sample of 300k words (vertical
dashed lines), extrapolated to $13\times$ the initial sample size.
Expected values were obtained by subsampling the full data set
(about 9$\times$ the initial sample size) without
replacement. (a) Estimator \dengnohat{}. (b) ZTNB.}
 \label{fig:dickens}
\end{figure}

\subsection{Followers in a social network}

We also applied the estimator \dengnohat{} to predict the number of active Twitter
users, which have $r$ or more followers, based on a small sample of
``(follower, followed)'' relations.  We obtained a data set from the
Social Computing Data Repository \citep{Zafarani2009}. This data set
contains 11.3M users and over 85.3M following relationships, which
form edges in this social network. We randomly sampled 5M edges as an initial
sample and used this to estimate the number of users with at least $r$
followers in larger sets of following relationships.
Estimates using \dengnohat{} show high accuracy when we extrapolated to
$5\times$ the size of the initial sample, as can be seen in
Figure~\ref{fig:twitter}. For example, when the total number of
following relations is 25M, we should see roughly 1.2M
individuals with at least 2 followers (Table~S4). Our prediction of just around
1.3M is off by $3\%$.
Accuracy decreases for larger extrapolations. The entire dataset
contains 2.8M users, each of which has at least $r=2$ followers. Our
estimator predicts 3.1M, an overestimate of around $10\%$. Interestingly,
accuracy does not rapidly worsen with $r$
and seems to remain high for values of $r$ up to 100
(Figure~S3b), consistent with our results on the
linguistic data set.

For both the Dickens and the Twitter applications, the error from ZTNB
is substantially higher than from our estimator (Figure~\ref{fig:dickens}, 
\ref{fig:twitter}). At the same time, the estimates from ZTNB are less sensitive to $r$
than those of \dengnohat{} (Figure~S3).

\begin{figure}[t!]
\centerline{\includegraphics{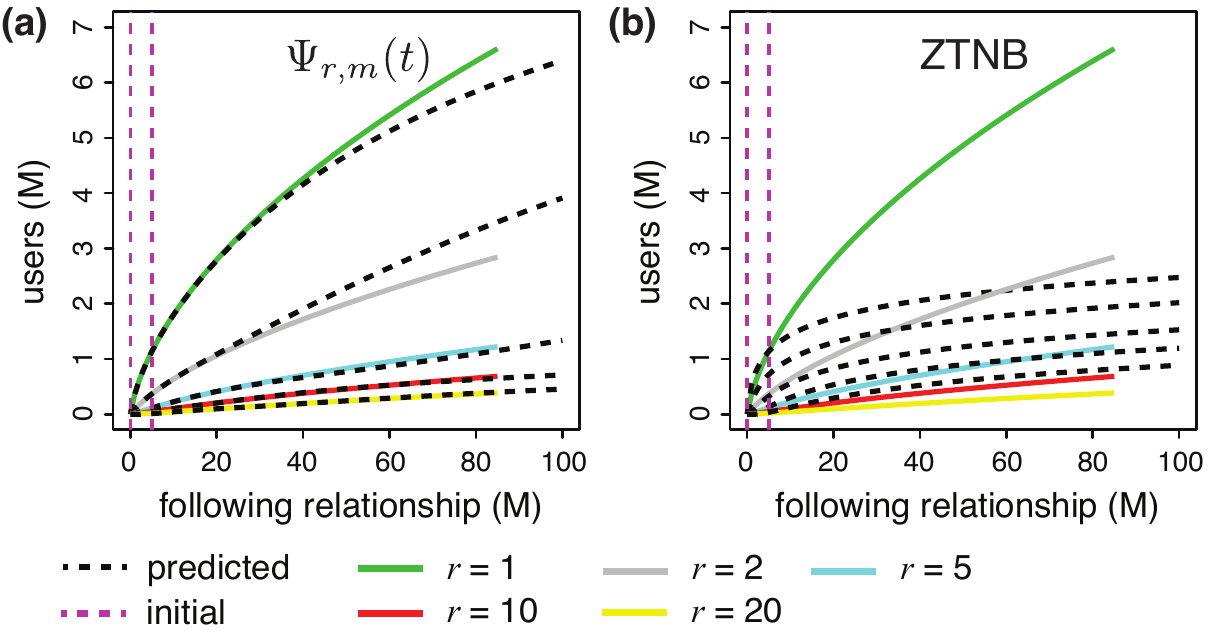}}
 \caption{Twitter users with at least $r$ followers in
     a sample of following relations.
   Estimates are based on an initial sample of 5M relations (vertical
   dashed lines), extrapolated to 20$\times$ initial sample
   size. Expected values were obtained by subsampling without
   replacement the full data set, which is about 17$\times$ the size
   of the initial sample. (a) The estimator \dengnohat{} and (b)
   ZTNB.}\label{fig:twitter}
\end{figure}

\subsection{Depth of coverage in DNA sequencing experiments}

To evaluate our approach on a larger scale, we applied our estimator to predict
the number of base pairs in the human genome that will be represented
at least $r$ times in a sequencing data set. In genomics terminology,
these are the positions in the genome covered by at least $r$
sequenced reads, or the positions with coverage depth at least $r$.
Coverage depth is critical in genetics studies, for example in
detecting SNPs, where candidate SNPs with low coverage depth are often
discarded. Knowing the distribution of coverage depth can help
researchers in experimental design, informing the total amount of
required sequencing in order to attain sufficient coverage over an
acceptable number of genomic sites \citep{zou2016quantifying}.

We downloaded four publicly available DNA sequencing experiments
(accession id SRX202787, \\
SRX205367, SRX204160 and SRX151616)
from NCBI to evaluate the performance of \dengnohat{}. Datasets were preprocessed
to obtain the number of genomic sites $N_j$ covered by
exactly $j$ reads (See Section~S7 for the details
of the preprocessing procedures).
We estimated the number of sites that would attain minimum
required depth as sequencing continues based on counts $N_j$.
Figure~\ref{fig:dna}a shows the curves for estimated values
\dengnohat{} for multiple values of $r$, along with the actual
expected values obtained by repeated subsampling from the full data
set.  Figure~\ref{fig:dna}b presents the same information for
estimates based on the ZTNB. This data set is sufficiently large to
reveal the inflection points in the curves when $r > 1$. Estimates
from \dengnohat{} closely track the true values. Even extrapolating up
to 100 times, the relative error is less than $5\%$ for various $r$
(Table~S5). The ZTNB, on the other hand,
overestimates $\E[S_1(t)]$ and then underestimates for other values of
$r$.
On the remaining three data sets, both methods show varying accuracy,
which is almost always higher for \dengnohat{}
(Figure~S4a -- S4f).

\begin{figure}[t!]
\centerline{\includegraphics{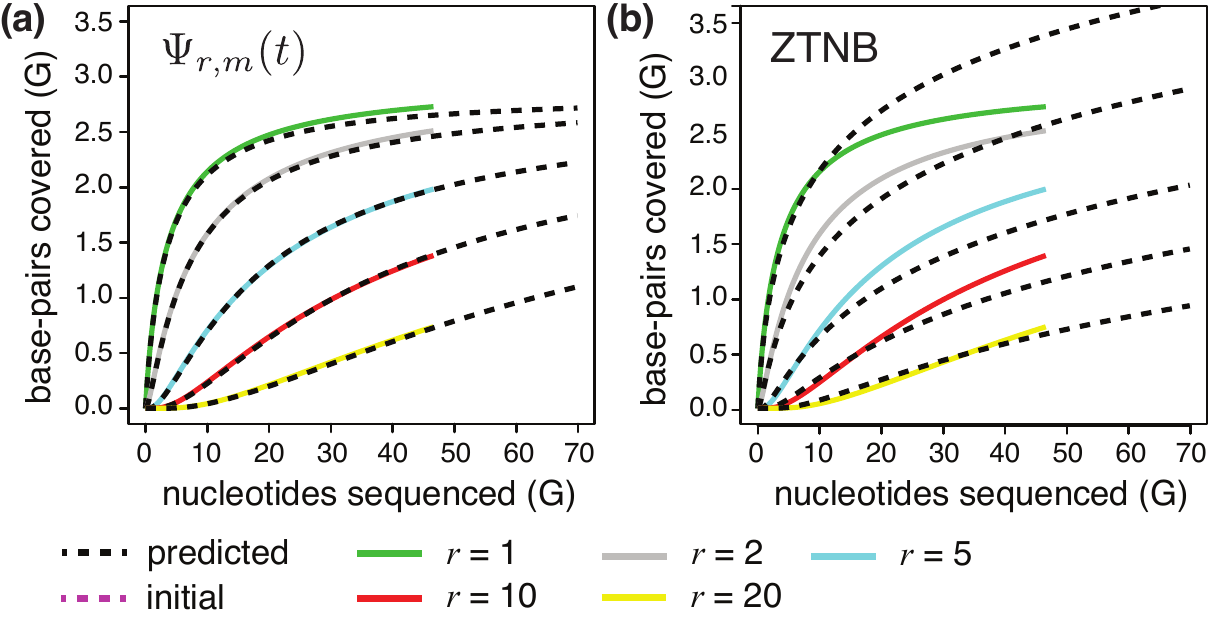}}
 \caption{Base pairs covered at least $r$ times in
     a DNA sequencing data set.
   Estimates are based on an initial sample of 5M reads (500M
   nucleotides), extrapolated to
   more than 160 $\times$ the initial sample size.  Expected values
   were from subsampling the full data set without replacement, which
   is around 107 $\times$ the size of the initial sample. Initial
   sample size is not indicated as the small size would not be
   visible.  Estimates made using (a) the estimator \dengnohat{} and
   (b) ZTNB.}\label{fig:dna}
\end{figure}

\section{Discussion}

We introduced a new approach to estimate the number of species that
will be represented at least $r$ times in a sample. The nonparametric
estimators obtained by our approach are universal in the sense that
they apply across values of $r$ for a given population. We have shown
that these estimators have favorable properties in theory, and also
give highly accurate estimates in practice. Accuracy remains high for
large values of $r$ and for long-range extrapolations. This approach
builds on the theoretical nonparametric empirical Bayes foundation
of \cite{good1956number}, providing a practical way to compute
estimates that are both accurate and stable.

The foundation for our approach is a relation between the
$r$-species accumulation curve $\E[S_r(t)]$ and the
$(r-1)^{\mathrm{th}}$ derivative of the average discovery rate
$\E[S_1(t)] / t$. This relation characterizes $\E[S_r(t)]$
directly, avoiding the summation of $\E[N_j(t)]$ estimates.
Clearly any estimator for either $\E[S_r(t)]$ or $\E[N_r(t)]$ provides
a means of estimating both quantities. By definition $S_r(t)$ is the
sum of $N_j(t)$ for $j\geq r$. Similarly $N_r(t)$ can be written as
$S_r(t) - S_{r+1}(t)$. We prefer to work with $S_r(t)$ because
$\E[S_r(t)]$ is an increasing function---an property that is
extremely useful for identifying problems during estimator
construction (Section~\ref{sec:algorithm}).

We use rational functions to approximate $r$-SAC. The advantages of
RFA stem from increased freedom to describe functions or to constrain
how those functions are estimated.  The coefficients of an
approximating rational function are usually determined in a way that
allows them to best fit observed data. The choice of forms for
rational functions, on the other hand, can be independent of the data
and can be determined by prior knowledge of the target function we
seek to approximate. In this work we use a class of rational functions
$P_{m-1}(t) / Q_m(t)$ with numerator degree $m-1$ and denominator
degree $m$ to mimic the behavior of the average discovery rate, which
is close to $L / t$ for large $t$. Other forms of
rational functions, in particular example $P_{m}(t) / Q_m(t)$ and
$P_{m+1}(t) / Q_m(t)$, were used when those forms made sense
\citep{daley2013predicting,deng2015}.

%

Our empirical accuracy evaluations were based on applications in which
the underlying data sets can be considered large compared to
traditional applications from ecology.  In particular, for modern
biological sequencing applications, samples are frequently in the
millions, and the scale of the data could be different by orders of
magnitude. These large-scale applications present new challenges to
traditional capture-recapture statistics, and call for methods which
can integrate high-order moments to accurately characterize the
underlying population. We generalized the classical study of
estimating a species accumulation curve and propose a nonparametric
estimator that can theoretically leverage any number of moments. We
believe both this generalization and the associated methodology
suggest possible avenues for practical advances in related estimation
problems.


\bibliographystyle{Chicago}
\bibliography{biblio}

\end{document}


\date{}
\maketitle

\section{Proofs}

\subsection{Proof of Theorem~1}
\label{sec:proofmainthm}

\begin{theorem}
Let $S_r(t)$ denote the number of species represented at least $r$
times for $t$ units of time.  For any positive integer $r$,
\[
\E[S_r(t)] = \frac{(-1)^{r-1} t^r}{(r-1)!}
\diff{\left(\frac{\E[S_1(t)]}{t}\right)}{r-1}. 
\]
\end{theorem}

\begin{proof}
The theorem is trivially satisfied for $r=1$, so we assume $r \geq
2$. By equation (1) and linearity of expectation,
\[
\E[S_r(t)] = \E[S_1(t)] - \sum_{j=1}^{r-1} \E[N_{j}(t)].
\]
Equation~(2) shows how to eliminate the $\E[N_j(t)]$
and write $\E[S_r(t)]$ in terms of derivatives of $\E[S_1(t)]$:
\begin{align}\label{eqn:sum}
\E[S_r(t)]
&= \E[S_1(t)] -
\sum_{j=1}^{r-1} \frac{(-1)^{j-1}t^j}{j!} \diff{\E[S_1(t)]}{j}\notag\\
&= \sum_{j=0}^{r-1} \frac{(-1)^{j}t^j}{j!} \diff{\E[S_1(t)] }{j}.
\end{align}
Extracting a factor of $t^r$ from the summands results in
\[
\E[S_r(t)] =
t^r\sum_{j=0}^{r-1} \frac{(-1)^{j}}{j!}\frac{1}{t^{r - j}} \diff{\E[S_1(t)] }{j}.
\]
Note that $1/t^{r-j}$ can be expanded as
\[
\frac{1}{t^{r-j}} =
\frac{(-1)^{r - 1 - j}}{(r - 1 - j)!}\diff{\left(\frac{1}{t}\right)}{r-1-j},
\]
which suggests a substitution in equation~\eqref{eqn:sum}. Eliminating
the factor $1/t^{r-j}$ and introducing the $(r-1-j)^{\mathrm{th}}$
derivative of $1/t$, we arrive at:
\begin{align*}
t^r\sum_{j=0}^{r-1} \frac{(-1)^{j}}{j!}\frac{1}{t^{r - j}}\diff{\E[S_1(t)]}{j}
&= t^r\sum_{j=0}^{r-1} \frac{(-1)^{j}}{j!}\frac{(-1)^{r-1-j}}{(r-1-j)!}
      \diff{\E[S_1(t)]}{j}\diff{\left(\frac{1}{t}\right)}{r-1-j} \\
&= t^r\sum_{j=0}^{r-1} \frac{(-1)^{r-1}}{(r-1)!}
      {r-1 \choose j}\diff{\E[S_1(t)]}{j}\diff{\left(\frac{1}{t}\right)}{r-1-j}\\
&= \frac{(-1)^{r-1}t^r}{(r-1)!}\sum_{j=0}^{r-1}
      {r-1 \choose j}\diff{\E[S_1(t)]}{j}\diff{\left(\frac{1}{t}\right)}{r-1-j}.
\end{align*}
Finally, by noticing that the above summation has the form of the
general Leibniz rule,
\begin{align*}
\E[S_r(t)]
&= \frac{(-1)^{r-1}t^r}{(r-1)!}
\sum_{j=0}^{r-1}
{r-1 \choose j} \diff{\E[S_1(t)]}{j}\diff{\left(\frac{1}{t}\right)}{r - 1 - j}\\
&= \frac{(-1)^{r-1}t^r}{(r-1)!}
\diff{\left(\frac{\E[S_1(t)]}{t}\right)}{r-1}.\qedhere
\end{align*}
\end{proof}


\subsection{Results for proving existence and uniqueness of the estimator}
\label{subsec:thmexist}

In order to prove Theorem~2, we require two lemmas. In
Lemma~1, we establish the power series
representation for the average discovery rate. In
Lemma~2 we give a sufficient condition for the
existence of the Pad\'{e} approximant to $\phi(t)$, which is the power
series estimator defined in equation~(5) for
the average discovery rate.

\begin{lem}
If $0 < t < 2$, then
\[
\frac{\E[S_1(t)]}{t} = \sum_{i=0}^\infty (-1)^{i}(t - 1)^{i}\E[S_{i+1}].
\]
\end{lem}

\begin{proof}
The expected number of species for $t$ units of time is
\[
\E[S_1(t)] = L\int (1 - e^{-\lambda t})dG(\lambda) = L\int (1 - e^{-\lambda (t-1)}e^{-\lambda})dG(\lambda).
\]
Replacing $\exp(-\lambda(t-1))$ with its power series yields
\begin{align*}
\E[S_1(t)]
&= L\int \left(1 - \sum_{i=0}^\infty \frac{(-\lambda(t-1))^i}{i!}e^{-\lambda}\right)dG(\lambda)  \\
&= L\int\left(1-e^{-\lambda}\right)dG(\lambda) - \sum_{i=1}^\infty (-1)^i (t-1)^i L\int\frac{\lambda^i}{i!}e^{-\lambda}dG(\lambda).
\end{align*}
Note that the expected number of species in the initial sample is
\[
\E[S_1] = L\int\left(1-e^{-\lambda}\right)dG(\lambda)
\]
and the expected number of species observed $j$ times in the initial
sample is
\[
\E[N_j] = L\int\frac{\lambda^je^{-\lambda}}{j!}dG(\lambda).
\]
Therefore, the expected value of $S_1(t)$ can be expressed as
\begin{equation}\label{eqn:GT-ps}
\E[S_1(t)] = \E[S_1] + \sum_{i=1}^\infty(-1)^{i - 1}(t - 1)^i\E[N_i].
\end{equation}
The expansion of $1/t$, at $t = 1$, is
\[
\frac{1}{t} = \frac{1}{1 + t - 1} = 1 - (t - 1) + (t - 1)^2 - \cdots
\]
for $0 < t < 2$.
Replacing both $\E[S_1(t)]$ and $t$ with their power series yields
\begin{align*}
\E[S_1(t)]/t &= \left( \E[S_1] + \sum_{i=1}^\infty(-1)^{i - 1}(t - 1)^i\E[N_i] \right)
   \left( \sum_{i=0}^\infty (-1)^i(t - 1)^i \right)\notag\\
&= \E[S_1] + \sum_{i=1}^\infty (-1)^i (t - 1)^i \left(\E[S_1] - \sum_{j=1}^i \E[N_j]\right)\\
&= \sum_{i=0}^\infty (-1)^{i}(t - 1)^{i}\E[S_{i+1}].
\end{align*}
\end{proof}


The following lemma shows the
existence of the Pad\'{e} approximant to $\phi(t)$. The conditions used in 
the lemma deviate from the classic result 
\citep[Theorem~3.3]{gragg1972pade} in two ways. 
First, the condition 
$\Delta_{m,m} \neq 0$ is added to ensure the leading coefficient in the denominator
in the Pad\'{e} approximant is nonzero. This condition is not required for proving the 
existence of the Pad\'{e} approximant, which is shown in the first part of the proof. 
The condition is used to prove $b_m \neq 0$, which 
is needed in Theorem~2.
Second, entries in $\Delta_{i,j}$ are not exactly coefficients of $\phi(t)$, but their
absolute values. In the second part of the proof, we show two forms of determinants
are equal. 

\begin{lem}\label{lem:existpade}
If the determinants $\Delta_{m-1, m}$ and $\Delta_{m, m}$ are nonzero,
there exist real numbers $a_i$ and $b_j$ for $i = 0, \ldots, m - 1$
and $j = 1, 2, \ldots, m$, with $b_m\neq 0$, such that the rational function
\[
\frac{P_{m-1}(t)}{Q_m(t)} = \frac{a_0 + a_1(t - 1) + \cdots + a_{m-1}(t-1)^{m-1}}{1 + b_1(t-1) + \cdots + b_m(t-1)^m}
\]
satisfies
\[
\phi(t) - \frac{P_{m-1}(t)}{Q_m(t)} = O\left((t - 1)^{2m}\right),
\]
where the determinant $\Delta_{i, j}$ is defined in (7)
and the truncated power series $\phi(t)$ is defined in
equation~(5).  Further, all $a_i$ and $b_j$ are uniquely determined by $S_1,S_2,\ldots,S_{2m}$.
\end{lem}

\begin{proof}
For the sake of convenience, we use $b_0$ to denote $1$. Since the
denominator $Q_m(t)$ is nonzero, by multiplying $Q_m(t)$ on both sides
of expression~(8), the expression becomes
\[
P_{m-1}(t) -  \phi(t)Q_m(t) = O\left((t - 1)^{2m}\right).
\]
Substitute the expression~(5) for $\phi(t)$.
The above condition is equivalent to two sets of system equations:
\begin{align}\label{eqn:regpade1}
b_m(-1)^0S_1 + b_{m-1}(-1)^1S_2 + \cdots + b_0(-1)^mS_{m+1} &= 0 \nonumber \\
b_m(-1)^1S_{2} + b_{m-1}(-1)^2S_{3} + \cdots + b_0(-1)^{m+1}S_{m+2} &= 0 \nonumber \\
\vdots \\
b_m(-1)^{m-1}S_{m} + b_{m-1}(-1)^mS_{m+1} + \cdots + b_0(-1)^{2m-1}S_{2m} &= 0, \nonumber
\end{align}
and
\begin{align}\label{eqn:regpade2}
a_0 &= (-1)^0S_1 \nonumber \\
a_1 &= (-1)^1S_2 + b_1(-1)^0S_1 \nonumber \\
\vdots \\
a_{m-1} &= (-1)^{m-1}S_m + \sum_{i=1}^{m-1} b_i(-1)^{m-i-1}S_{m-i}. \nonumber
\end{align}
Clearly once $b_j$ are determined for $j = 1, 2, \ldots, m$, the values of $a_j$ for
$j = 0, 1, \ldots, m-1$ can be calculated through the linear system
equations~\eqref{eqn:regpade2}. Thus the existence of the rational function
$P_{m-1}(t) / Q_m(t)$ solely depends on the solution
of the linear system equations~\eqref{eqn:regpade1}.
Using Cramer's rule, the linear system
equations~\eqref{eqn:regpade1} have unique solutions if and only if the following
determinant is nonzero
\begin{equation}\label{eqn:determinantoriginal}
\begin{vmatrix}
  (-1)^0S_{1} & (-1)^{1}S_{2} & \dots  & (-1)^{m-1}S_{m} \\
  (-1)^1S_{2} & (-1)^{2}S_{3} & \dots  & (-1)^{m}S_{m+1} \\
  \vdots & \vdots &  \ddots & \vdots \\
  (-1)^{m-1}S_{m} & (-1)^{m}S_{m+1} & \dots  & (-1)^{2m-2}S_{2m-1}
\end{vmatrix}_{m\times m}.
\end{equation}

We can simplify the form of this determinant.  Using the permutation
definition of a determinant, the above determinant can be expressed as
\[
\sum_{\sigma} sgn(\sigma)~a_{1,\sigma(1)}a_{2,\sigma(2)}\cdots a_{m,\sigma(m)},
\]
where $\sigma$ is a permutation for the set $\{1, 2, \ldots, m\}$ and $a_{i,j}$ is
$(-1)^{i+j}S_{i+j+1}$. We can write the sum as
\[
\sum_{\sigma} sgn(\sigma)~a_{1,\sigma(1)}a_{2,\sigma(2)}\cdots a_{m,\sigma(m)} =
\sum_{\sigma} sgn(\sigma)~(-1)^{\sum_{i=1}^m i + \sigma(i)}
| a_{1,\sigma(1)} | | a_{1,\sigma(1)} |\cdots | a_{m,\sigma(m)} |.
\]
Since the sum of $\sigma(i)$ is equal to the sum of $i$, we have
\[
(-1)^{\sum_{i=1}^m i + \sigma(i)} = 1.
\]
As a result, the determinant \eqref{eqn:determinantoriginal} is equal to
\[
\Delta_{m-1, m} =
\begin{vmatrix}
  S_{1} & S_{2} & \dots  & S_{m} \\
  S_{2} & S_{3} & \dots  & S_{m+1} \\
  \vdots & \vdots &  \ddots & \vdots \\
  S_{m} & S_{m+1} & \dots  & S_{2m-1}
\end{vmatrix}_{m\times m}.
\]
Therefore, the Pad\'{e} approximant to the power series $\phi(t)$
exists, provided $\Delta_{m-1, m} \neq 0$. From system equations
\eqref{eqn:regpade1} and \eqref{eqn:regpade2}, it is clear that the
$a_i$ and $b_j$ are uniquely determined by $1 \leq S_r \leq 2m$. 
Note that $b_m$ can be solved as
$(-1)^m \Delta_{m,m} / \Delta_{m-1, m}$, which is nonzero by
assumptions.
\end{proof}


Based on Lemma 1 and 2, we
propose our estimator for $r$-SAC in Theorem~2 in the
main text with a proof.  Recall that in the proof, we assume that the
denominators of all rational functions of interest have simple roots.
Here we release this constraint so that the multiplicity of a root of the
denominator in a rational function could be greater than $1$.  In
general, the Pad\'{e} approximant $P_{m-1}(t) / Q_{m}(t)$ can be
represented as
\[
P_{m-1}(t) / Q_m(t) = \sum_i \left( \frac{c_{i1}}{t - x_i} +
\frac{c_{i2}}{(t - x_i)^2} + \cdots + \frac{c_{ik_i}}{(t - x_i)^{k_i}}\right),
\]
where $x_i$ is the root of multiplicity $k_i$ for $Q_m(t) = 0$ and the sum of $k_i$ is $m$.
So the estimator \dengnohat{} can be expressed as
\[
\dengnohat{} = \sum_i {r-1 \choose r-1}\frac{c_{i1}t^r}{(t - x_i)^r} +
{r \choose r-1}\frac{c_{i2}t^r}{(t - x_i)^{r+1}} + \cdots +
{r - 2 + k_i \choose r-1}\frac{c_{ik_i}t^r}{(t - x_i)^{r-1+k_i }}.
\]
In particular, if $k_i = 1$ for all $i$, {\it i.e.} all roots of
$Q_m(t) = 0$ are simple, the estimator \dengnohat{} can be expressed
as (9).  For real applications, we rarely met
$Q_m(t)$ with a repeat root. Therefore for rest of the paper, we
assume denominators of RFA have only simple roots.  We use
\dengnohat{} to represent the estimator of the form
\[
\dengnohat = \sum_{i=1}^m c_i\left( \frac{t}{t - x_i} \right)^r.
\]

\subsection{Proofs of Proposition~1}
\label{sec:estimatorproperty}

\begin{prop}
The following hold for the estimator \dengnohat{} of Theorem~2:\\
(i) The estimator \dengnohat{} is unbiased for $\E[S_r(t)]$ at $t = 1$ for $r \leq 2m$.\\
(ii) The estimator \dengnohat{} converges as $t$ approaches infinity. In particular,
\[
\lim_{t\to\infty} \dengnohat = \frac{\Delta_{m-1, m+1}}{\Delta_{m, m}}.
\]
(iii) The estimator \dengnohat{} is strongly consistent as the
initial sample size $N$ goes to infinity.
\end{prop}

\begin{proof}
  {\it (i)} The result is directly derived from Theorem~2.

  \noindent
  {\it (ii)} When $t$ goes to infinity, the ratio $t /(t - x_i)$ goes
  to $1$, so \dengnohat{} converges to $\sum c_i$ for all $r$. Recall that the
  Pad\'{e} approximant to the power series $\phi(t)$ can be expressed
  \[
  P_{m-1}(t) / Q_m(t) = \sum_{i=1}^m \frac{c_i}{t - x_i}.
  \]
  by equation (13).
  Thus $\sum c_i$ is equal to the ratio $tP_{m-1}(t) / Q_m(t)$ as $t$
  goes to infinity. Using the determinant representation of
  $P_{m-1}(t) / Q_m(t)$ in equation~(6), we
  immediately obtain that the sum of $c_i$ is equal to
  $\Delta_{m-1,m+1}/\Delta_{m,m}$.

  \noindent
  {\it (iii)} Let $A_N$ denote the event that $S_{k} < L$ after sampling
  $N$ individuals and
  \[
  A(i.o.) = \{A_N~\mbox{occurs for infinitely many $N$}\}.
  \]
  In the following we show $\Pr(A(i.o.)) = 0$; with probability 1 only
  a finite number of events $A_N$ occur. Let $p_i > 0$ be the
  probability that a randomly sampled individual belongs to species
  $i$.  According to our modeling assumptions, individuals are sampled
  \textit{i.i.d.}
  given the value of $N$. The probability that
  species $i$ is observed fewer than $k$ times in the sample is
  \begin{equation}\label{qijn}
    q_{i,k,N} =
    \begin{cases}
      \sum_{\ell=0}^{k-1}{N \choose \ell}p_i^\ell(1 - p_i)^{N-\ell}, & \text{ if } N \geq k, \\
      1, & \text{otherwise}.
    \end{cases}
  \end{equation}
  Since
  \begin{align*}
    \Pr(A_N) & = \Pr\big(\cup_{i=1}^L\mbox{ species $i$ is observed fewer than $k$ times}\big) \\
    & \leq \sum_{i=1}^L \Pr\big(\mbox{species $i$ is observed fewer than $k$ times}\big)= \sum_{i=1}^{L}q_{i,k,N},
  \end{align*}
  the sum of $\Pr(A_N)$ satisfies
  \[
  \sum_{N=1}^\infty \Pr(A_N) \leq \sum_{N=1}^\infty \sum_{i=1}^{L}q_{i,k,N} =  \sum_{i=1}^{L} \sum_{N=1}^\infty q_{i,k,N}.
  \]
  For any given $i$, one can verify
  \[
  \sum_{N=1}^\infty q_{i,k,N} < \infty.
  \]
  As the sum of $\Pr(A_N)$ is finite, by the Borel-Cantelli Lemma
  $\Pr(A(i.o.)) = 0$. By definition $S_{k} \leq S_{k-1} \cdots \leq
  S_{1} \leq L$. So if we consider $k=2m$, with probability 1 there
  exists an $N_0$ such that $S_{j} = L$ for all $j \leq 2m$ when $N >
  N_0$.

  Next we show that the estimator \dengnohat{} is $L$ whenever $S_{j}
  = L$ for $j \leq 2m$ and $t > 0$. In the special case $m = 1$, the
  estimator \dengnohat{} can be checked directly:
  \[
  \dengnohat{} = \frac{S_{1}^2}{S_{2}}\left(\frac{t}{t + (S_{1} - S_{2}) / S_{2}} \right)^r = L.
  \]
  For $m > 1$, we must establish that the Pad\'{e} approximant
  $P_{m-1}(t)/Q_{m}(t)$ reduces to $L/t$, 
  so that
  \[
  \dengnohat{} = \frac{(-1)^{r-1} t^r}{(r-1)!}\diff{\left(\frac{L}{t}\right)}{r-1} = L.
  \]

  Let $P_{m-1}(t)/Q_{m}(t)$ denote the Pad\'{e} approximant to
  $\phi(t)$ with numerator $m-1$ and denominator $m$
  \[
  \frac{P_{m-1}(t)}{Q_m(t)} = \frac{a_0 + a_1(t - 1) + \cdots + a_{m-1}(t-1)^{m-1}}{1 + b_1(t-1) + \cdots + b_m(t-1)^m},
  \]
  where $a_l$ and $b_k$ are real numbers. Recall that the Pad\'{e}
  approximant exists exactly when both systems of equations
  \eqref{eqn:regpade1} and \eqref{eqn:regpade2} have solutions with
  $b_0 = 1$. When $S_j = L$ for $j \leq 2m$, the system
  \eqref{eqn:regpade1} degenerates to one equation
  \begin{equation}\label{eqn:degcondition}
    b_m - b_{m-1} + \cdots + b_0(-1)^m = 0.
  \end{equation}
  Given any values for $b_1, b_2,
  \ldots, b_{m-1}$, the values of $a_0, a_1, \ldots, a_{m-1}$ and
  $b_m$ are uniquely determined by
  \eqref{eqn:regpade2} and \eqref{eqn:degcondition}.
  Based on the system of equations \eqref{eqn:regpade2}, the polynomial $P_{m-1}(t)$ can be expressed
  as
  \begin{multline*}
    a_0 + a_1(t - 1) + \cdots + a_{m-1}(t - 1)^{m-1} \\
    = b_0L + \left(-b_0L + b_1L\right)(t-1) +  \cdots +
    \left(b_0(-1)^{m-1}L + \sum_{k=1}^{m-1}b_k(-1)^{m-1-k}L\right)(t-1)^{m-1}.
  \end{multline*}
  Terms that include the factor $b_0$ can be collected as
  \begin{align*}
    b_0\left(L - L(t - 1) + \cdots + (-1)^{m-1}L(t - 1)^{m-1}\right)
    & = b_0L\left(\frac{1 - (-1)^m(t-1)^m}{1 - (-1)(t - 1)} \right)\\
    & = \frac{L}{t}\left(b_0 - b_0(-1)^m(t-1)^m \right).
  \end{align*}
  Similarly, for terms including factors $b_k$, with $k = 1, 2, \ldots, m-1$,
  \begin{multline*}
    b_k\left(L(t -1)^k - L(t-1)^{k+1} + \cdots + (-1)^{m-k-1}L(t-1)^{m-1}\right)\\
    = b_k(t-1)^kL\left(\frac{1 - (-1)^{m-k}(t-1)^{m-k}}{t}\right)
    = \frac{L}{t}\left( b_k(t-1)^k - (-1)^{m-k}b_k(t-1)^m\right).
  \end{multline*}
  Thus,
  \begin{multline*}
    a_0 + a_1(t - 1) + \cdots + a_{m-1}(t - 1)^{m-1} \\
    =\frac{L}{t}\left(b_0 + b_1(t-1) + \cdots + b_{m-1}(t-1)^{m-1} -
    \sum_{k=0}^{m-1}b_k(-1)^{m-k}(t-1)^m\right).
  \end{multline*}
  Equation \eqref{eqn:degcondition} provides the following simplification:
  \[
  \sum_{k=0}^{m-1} b_k(-1)^{m-k}(t-1)^m = -b_m(t-1)^m.
  \]
  Therefore, the numerator polynomial $P_{m-1}(t)$ can be written
  \[
  a_0 + a_1(t - 1) + \cdots + a_{m-1}(t - 1)^{m-1}
  = \frac{L}{t}\left( b_0 + b_1(t-1) + \cdots + b_m(t-1)^m \right).
  \]
  The rational function $P_{m-1}(t) / Q_m(t)$ is then
  \[
  \frac{P_{m-1}(t)}{Q_m(t)} = \frac{L}{t},
  \]
  which holds for any choices of $b_1, b_2, \ldots, b_{m-1}$. As a
  result, the estimator \dengnohat{} becomes
  \[
  \dengnohat{} = \frac{(-1)^{r-1} t^r}{(r-1)!}\diff{\left(\frac{L}{t}\right)}{r-1} = L
  \]
  whenever $S_1 = S_2 = \cdots = S_{2m} = L$ and $t > 0$.

  Let $B_N$ be the event $\dengnohat{} \neq L$ in a random sample of
  size $N$ and $B(i.o.)$ to be the event that $B_N$ occurs for
  infinite many $N$. Since 
  $B(i.o) \subset A(i.o)$ and $\Pr(A(i.o))=0$, it follows that $\Pr(B(i.o.)) = 0$.
  Therefore $\dengnohat{} \overset{a.s}{\to} L$ as sample size $N$
  goes to infinity for any $t > 0$.
\end{proof}

\subsection{Proof of Proposition~2}
\label{sec:pacman}

\begin{prop}
If $\Rp(x_i) < 0$ for $1 \leq i \leq m$, then \dengnohat{} is bounded for any $t \geq 0$
and $r \geq 1$, where $\Rp(x_i)$ is the real part of the complex number $x_i$.
Further, $\dengnohat{}\rightarrow 0$
as $r\rightarrow \infty$ for any $0\leq t < \infty$.
\end{prop}

\begin{proof}
If $\Rp(x_i) < 0$ then for any given $t\in[0, \infty)$ and any
$x_i$,
\[
\| t - x_i \|^2 = t^2 - (x_i + \bar{x}_i) t + x_i\bar{x}_i > t^2.
\]
Therefore $\|t\slash(t - x_i)\| < 1$, and
\[
\left\| \dengnohat \right\| = \left\|\sum_{i=1}^m c_i\left(\frac{t}{t - x_i}\right)^r \right\|
\leq \sum_{i=1}^m \left\| c_i \right\| \left\|\left(\frac{t}{t - x_i}\right)^r \right\| \leq \sum_{i=1}^m \|c_i\|.
\]
So \dengnohat{} is bounded by the sum of absolute $c_i$ for any
$r$. As $r$ approaches infinity, we obtain
\[
\lim_{r\to\infty}\left\|\sum_{i=1}^m c_i\left(\frac{t}{t - x_i}\right)^r \right\|
~ \leq ~
\sum_{i=1}^m \|c_i\| \lim_{r\to\infty}\left\|\frac{t}{t - x_i}\right\|^r
~ = ~ 0.\qedhere
\]
\end{proof}


\section{Generalizing estimators of SAC to $r$-SAC}
\label{sec:boosting}

The formula~(3) in Theorem~1
provides a powerful tool to generalize existing
estimator for the SAC and obtain estimators for the $r$-SAC. All estimators
below are derived through substituting an estimator of
SAC for $\E[S_1(t)]$ in the formula~(3) except for the logseries estimator,
which is implied by \cite{fisher1943relation}.
The first two parametric estimators have been discussed in Section~2.
We directly give the result without derivation.

\paragraph{Zero-truncated Poisson estimator.}
The number of species $\E[S_r(t)]$ represented at least $r$ times can
be estimated as
\begin{equation}
\hat\E[S_r(t)] = \frac{S_1}{1 - e^{-\lambda}} \left(1 - \sum_{i=0}^{r-1} \frac{(\lambda t)^i}{i!}e^{-\lambda t}\right),
\end{equation}
where $\lambda$ satisfies
\[
\frac{\lambda}{1 - e^{-\lambda}} = \frac{\sum_{i=1}^\infty iN_i}{S_1},
\] and is the MLE
of zero-truncated Poisson distribution.

\paragraph{Zero-truncated negative binomial estimator.}
The number of species $\E[S_r(t)]$ represented at least $r$ times
can be estimated as
\begin{equation}
\hat\E[S_r(t)] =
\frac{S_1}{1 - (1 + \beta)^{-\alpha}} \left(1 - \sum_{i=0}^{r-1}\frac{\Gamma(i + \alpha)}{\Gamma(i + 1)\Gamma(\alpha)}\left(\frac{\beta t}{1+ \beta t}\right)^i\left(\frac{1}{ 1 + \beta t}\right)^{\alpha}\right),
\end{equation}
where $\alpha$ and $\beta$ are fit by an expectation-maximization algorithm,
with unobserved counts as the missing data.

\paragraph{The estimator of Boneh, Boneh and Caron (1998).} We refer
to this estimator as BBC. The expected number of species $\E[S_1(t)]$
represented at least once is estimated as
\[
\hat\E[S_1(t)] = S_1 + \sum_{i=1}^\infty N_ie^{-i}\left(1 - e^{-i(t-1)}\right) + U\left( e^{-N_1 / U} - e^{-N_1t / U}\right),
\]
where $t \geq 1$ and $U$ is the solution of the equation
\[
U \left(1 - e^{-N_1 / U}\right) = \sum_{i=1}^\infty N_ie^{-i},
\]
if the condition $N_1 > \sum_{i=1}^\infty N_ie^{-i}$ is satisfied \citep{boneh1998estimating}.
The number of species $\E[S_r(t)]$ represented at least $r$ times
can be estimated as
\begin{align*}
\hat\E[S_r(t)] &= \frac{(-1)^{r-1} t^r}{(r-1)!}
\diff{\left(\frac{S_1 + \sum_{i=1}^\infty N_ie^{-i}\left(1 - e^{-i(t-1)}\right) + U\left( e^{-N_1 / U} - e^{-N_1t / U}\right)}{t}\right)}{r-1} \\
&=
S_1 + \sum_{i=1}^\infty N_i \left(e^{-i} - \sum_{k=0}^{r-1} \frac{(i t)^k}{k!}e^{-i t}\right) + U \left(e^{-N_1 / U} - \sum_{i=0}^{r-1} \frac{(N_1t / U)^i}{i!}e^{-N_1t / U}\right).
\end{align*}

\paragraph{The estimator of Chao and Shen (2004).} We refer to this
estimator as CS. The expected number of species $\E[S_1(t)]$
represented at least once is estimated as
\[
\hat\E[S_1(t)] = S_1 + N_0 \left(1 - e^{-N_1(t-1) / N_0}\right),
\]
where $t \geq 1$ \citep{chao2004nonparametric}. Details for estimating $N_0$ can be found by
\cite{chao1992estimating}.
The number of species $\E[S_r(t)]$ represented at least $r$ times
can be estimated as
\begin{align*}
\hat\E[S_r(t)] &= \frac{(-1)^{r-1} t^r}{(r-1)!}
\diff{\left(\frac{S_1 + N_0 \left(1 - e^{-N_1(t-1) / N_0}\right)}{t}\right)}{r-1} \\
&=
S_1 + N_0 \left( 1 - \sum_{i=0}^{r-1} \frac{(N_1t / N_0)^i}{i!} e^{-N_1(t-1) / N_0} \right).
\end{align*}

\paragraph{The logseries estimator.}
The log series estimator for $\E[S_r(t)]$ is based on results by
\cite{fisher1943relation}.  According to the article, the expected
number of species represented $j$ times in $t$ units of sampling
effort can be estimated as
\[
\frac{\alpha}{j}x_t^j,
\]
where
\[
x_t = \frac{Nt}{\alpha + Nt}.
\]
The parameter $\alpha$ is estimated by solving the equation
\[
S_1 = \alpha \log \left(1 + \frac{N}{\alpha}\right).
\]
The number of species represented at least $r$ times in $t$ units of sampling effort is then
approximated by
\[
\hat\E[S_r(t)] = \sum_{i=r}^\infty \frac{\alpha}{i}x_t^i =  \alpha \int_0^{x_t} \frac{x^{r-1}}{x - 1} dx.
\]


\section{Implementation}
\label{sec:imp}

Algorithm 1 in the main text provides the basis to construct the estimator
\dengnohat{}. The following points address a few numeric details in the algorithm.

\begin{itemize}
\item Setting $m_\mathrm{max}$: if the largest non-zero count in the
  initial experiment is $N_j$, the value of $m$ can be chosen up
  to $\lfloor j / 2 \rfloor$. In our applications of interest this value can be
  extremely large (see Section~6), so we begin with
  a smaller value (we default to $m_\mathrm{max} = 10$).
\item Continued fraction: all the coefficients in the constructed continued fraction must
  be nonzero. If the $l$-th coefficient happens to be zero, only the
  first $l - 1$ nonzero coefficients are used when constructing
  Pad\'{e} approximants. Meanwhile, the value of $m_\mathrm{max}$ is
  adjusted to $\lfloor (l - 1) / 2 \rfloor$ corresponding to the
  number of nonzero terms in the continued fraction.
\item The time complexity for computing Pad\'{e} approximants is
  $O(m_\mathrm{max}^2)$.  Using the quotient-difference algorithm
  \citep{Rutishauser1954}, computing the coefficients of the degree
  $2m_\mathrm{max}$ associated continued fraction requires
  $O(m^2_\mathrm{max})$ time. After this is done, all desired Pad\'{e}
  approximants can be generated iteratively by evaluating the $2m$-th
  convergent of the continued fraction from $m = 1$ to
  $m_\mathrm{max}$~\cite[pp.~131]{baker1996pade}.
  This procedure takes $O(m_\mathrm{max}^2)$ time in total.
\item Defects in the Pad\'{e} approximant: a defect in the Pad\'{e} approximant
  $P_{m-1}(t) / Q_m(t)$ defines a pair, a pole $z_i$ and a nearby zero $y_j$. 
  Defects are common in Pad\'{e} 
  approximants~\citep[p.~48]{baker1996pade}.
  For example \cite{GILEWICZ1997199} shows that adding
  random errors to a geometric series could generate defects. 
  The defect causes the Pad\'{e} approximant to become unbounded in a
  neighborhood of its pole, but has little effect on values of the
  Pad\'{e} approximant outside the neighborhood. In order to construct
  robust estimators, we remove any defect found in $P_{m-1}(t) / Q_m(t)$. 
  
  For implementation, we first use R package \prog{polynom} (v1.3-8) 
  \citep{poly2014} to find all the poles $y_i$ and roots $z_j$ in the rational 
  function $P_{m-1}(t) / Q_m(t)$. We identify the pair ($z_i$, $y_j$) as a defect
  if their absolute difference is less than a threshold 
  ($0.001$ by default). If a defect $(z_i, y_j)$ is detected, we cancel out the factor 
  $(t - y_j)$ from numerator and $(t - z_i)$ from denominator. The simplified rational 
  function is then passed to the next step (partial fractions) and the obtained estimator is 
  evaluated using the criteria outlined in the main text.
\item Partial fraction decomposition: once all the roots $y_i$ of the denominator in the 
rational function $P_{m-1}(t) / Q_m(t)$ are determined, the coefficient $c_i$ associated with the root $y_i$ in the partial fraction can be easily obtained by the standard 
approach~\cite[pp.~276]{feller1968introduction}.
\item We use two criteria, the locations of poles and the monotonic shape
  of the SAC, to diagnose our estimator. In a strict sense, we desire
  that the constructed $r$-SAC be an increasing function, for every
  $r$. However, in our experience the SAC is sufficient: when it is
  monotone, so are the other $r$-SAC. The condition on the poles $x_i$
  is used to stabilize the estimator \dengnohat{}, as explained in
  Proposition~2. Estimators \dengnohat{} with larger
  $m$ tend to violate conditions in Proposition~2
  simply because they have more poles.
\item Monotonicity: the derivative of $\Psi_{1,m}(t)$ can be expressed as 
\[
\frac{\mathit{d}}{\mathit{dt}}\Psi_{1,m}(t) = -\sum_{i=1}^m \frac{c_ix_i}{(t - x_i)^2}.
\]
We evaluate the above function from $t = 1$ up to $100$ with step size $0.05$ to 
ensure that $\Psi_{1, m}(t)$ is an increasing function.
\end{itemize}

All the functionality has been implemented into an R package called
{\em preseqR} (v.3.1.2). It is available through CRAN at:

\smallskip\noindent
\url{https://CRAN.R-project.org/package=preseqR}

\section{The power series estimator of $r$-SAC}
\label{sec:ps}

In the main text, we have obtained the power series estimator
$\phi(t)$, which is defined by equation~(5) for
the average discovery rate. Instead of using the Pad\'{e} approximant
to the average discovery rate, we could directly apply $\phi(t)$ to
obtain a power series estimator for $r$-SAC.  In this section, we
derive the form of the power series estimator for $r$-SAC and discuss issues
associated with this power series estimator.

We use $\phi(t)$ and the formula~(3) to derive 
a power series estimator for $r$-SAC.  Substituting the power series
$\phi(t)$ for $\E[S_1(t)] / t$ in equation~(3) and taking
$r^\mathrm{th}$ derivatives leads to the desired form:
  \begin{equation}\label{eqn:psEst}
  \psEst{} = t^r \sum_{i=0}^\infty (-1)^i (t-1)^i {r - 1 + i \choose r - 1} S_{r + i}
  \end{equation}
In particular, when $r = 1$,
\begin{align*}
\Phi_1(t) &= t \sum_{i=0}^\infty (-1)^i S_{i+1} (t -1)^i \\
&= \sum_{i=0}^\infty (-1)^i S_{i+1} (t -1)^i + \sum_{i=0}^\infty (-1)^i S_{i+1} (t -1)^{i+1} \\
&= S_1 + \sum_{i=1}^\infty (-1)^{i-1}(t - 1)^iN_i,
\end{align*}
which is Good-Toulmin estimator \citep{good1956number}. 
Since the observed
$S_j$ are unbiased estimates for $\E[S_j]$, the following proposition
is straightforward.


\begin{prop}
  For any $r \geq 1$,
  \begin{equation}\label{eqn:psEst}
  \psEst{} = t^r \sum_{i=0}^\infty (-1)^i (t-1)^i {r - 1 + i \choose r - 1} S_{r + i}
  \end{equation}
  is an unbiased estimator for $\E[S_r(t)]$.
\end{prop}

In practice $\psEst{}$ is a truncated power series.
Even if we assume that the error in \psEst{} is acceptable when $t$ is
close to 1, we have no basis for anticipating acceptable behavior when
$t > 2$.
The radius of convergence for the expected value of \psEst{} is 1,
even though the quantity $\E[S_r(t)]$ itself is well defined. For any
fixed $r$ the factor
\[
{r - 1 + i \choose r - 1}
\]
in \eqref{eqn:psEst} has polynomial growth with $i$, 
and can add extreme weight
to each successive term $S_{i+r}$.
This could cause a large error in
the predictions when using the observations $S_{i+r}$ in the initial
sample to estimate its expectation $\E[S_{i+r}]$.

Let $j_{\max}$ be the largest $j$ such that observed $S_j > 0$ so
$S_{j'} = 0$ for all $j' > j_{\max}$. For any $r \leq j_{\max}$ the
estimator \psEst{} is a polynomial of degree $j_{\max}$.
This function will tend to oscillate up to a point where the term with
largest degree starts to dominate. Depending on whether the number
$j_{\mathrm{max}} - r$ even or odd, the estimate from \psEst{} might
become very large, or could even take negative values.

The problem with the truncated power series motivates us to use the
rational function approximation, in particular the Pad\'{e}
approximant.  The rational function approximation often gives better
estimation than truncated power series when estimating an alternating
power series.  They may even converge when the Taylor series fails to
converge. The following example is taken from
\cite{george1975essentials}:
\[
f(t) =
\left(\frac{1 + 2t}{ 1 + t}\right)^{1/2}
=
1 + \tfrac{1}{2}t - \tfrac{5}{8}t^2 + \tfrac{13}{16}t^3 - \tfrac{141}{128}t^4 + \cdots .
\]
This Taylor series diverges for $t > 1 / 2$, even though the function
$f(t)$ is well-defined over $[0, \infty)$. In particular, $f(t) \to
  \sqrt{2}$ as $t\to\infty$. Through a change of variable equating $t
  = w/(1 - 2w)$, we can rewrite
\begin{equation}\label{eqn:rfaexample}
f(t) = (1 - w)^{-1/2} = 1 + \tfrac{1}{2}w + \tfrac{3}{8}w^2 + \tfrac{35}{128}w^4 + \cdots .
\end{equation}
In terms of the original variable $t$, the expression
\eqref{eqn:rfaexample} induces a sequence of rational function
approximations to $f(t)$:
\[
1,~ \frac{1 + (5/2)t}{1 + 2t},~ \frac{1 + (9/2)t + (43/8)t^2}{(1+2t)^2},~ \ldots .
\]
As $t$ approaches infinity, these successive approximations for $f(t)$ converge to
%
%
$\sqrt{2}$.


\section{Best practice}
\label{sec:bestpractice}

For the best practice, we combine both our estimator \dengnohat{} and the ZTNB 
approach. Whenever samples are generated from a heterogeneous population, 
we use the estimator \dengnohat{}; otherwise, we switch to the ZTNB estimator.
The degree of heterogeneity in a population is measured by the 
coefficient of variation of $\lambda_i$ (equation~(16)).
For example, a small value of \cv{} indicates the relative species abundances
are close to each other, meaning that the population is close to a homogenous 
population. The problem with this strategy is that 
the value of $\lambda_i$ is unobservable. We can not directly assess the 
heterogeneity of a population.

We propose a heuristic approach to estimate the \cv{}. First, the initial sample 
is fitted to the ZTNB model. 
Let $k$ denote the estimated shape parameter in the model. 
The value $1 / \sqrt{k}$ is then used as an estimate for the \cv{}. In particular,
If the population follows a negative binomial distribution with the shape parameter $k$,
then the expected value of the $\cv{}$ is exactly $1 / \sqrt{k}$.
From the simulation results, 
we found that the shape parameter in the ZTNB approach is highly sensitive to the 
heterogeneity of a population (Figure~\ref{fig:cutoff_and_err}a).
When applying the ZTNB approach to a sample from a heavy-tailed population,
the estimated shape parameter is close to $0$. In contrast, for a sample from a
homogeneous population, the estimated shape parameter is large by orders of 
magnitude.

In practice, whenever the estimated shape parameter is less than $1$, or equivalently
the estimated \cv{} is greater than $1$, we use our
estimator \dengnohat{}; otherwise, we switch to the ZTNB estimator.
To examine this cutoff, we simulated samples from 
populations with \cv{}s from $0.1$ to $10$ using the NB model.
The mean of relative errors for \dengnohat{} drops quickly at $\cv{}=0.6$
(Figure~\ref{fig:cutoff_and_err}b).
At the point $\cv{} = 1$ ($k = 1$), where we draw the cutoff, 
the mean of relative errors decreases to the minimum value
0.002. As the \cv{} continues increasing from $1$, we observed that 
the prediction errors are less than
$0.05$ and become stable (Figure~\ref{fig:cutoff_and_err}c).
The result suggests that this cutoff $\cv{} = 1$ can 
select those cases that are in favor of our estimator and 
control the relative error in practice.


\section{Sample coverage}
\label{app:samplecov}

Assume a sample is drown from a population. Let $N$ denote the number of individuals in the sample.
Under our statistical assumption, individual are sampled \textit{i.i.d.} given the value of $N$.
The probability of an individual belonging to the species $i$ is
\[
p_i = \frac{\lambda_i}{\sum_{i=1}^L \lambda_i}.
\]
It is clear that the sum of $p_i$ is equal to 1.
The sample coverage is defined as
\begin{equation}\label{eqn:def_sample_cov}
C(p) = \sum_{i=1}^L p_i\delta_i,
\end{equation}
where $\delta_i$ is the indicator of whether or not species $i$ belongs to the sample. If species $i$ belongs to
the sample, the function value is $1$; otherwise it is $0$. For a random sample of size $N$, the probability of species $i$
not in the sample is $(1 - p_i)^N$. Therefore the expectation of sample coverage can be expressed as
\[
\E(C(p)) = \sum_{i=1}^L p_i \left(1- (1 - p_i)^N \right) = 1 - \sum_{i=1}^L p_i(1 - p_i)^N.
\]

\begin{prop}\label{prop:min_sample_cov}
For a random sample of fixed size $N$, the expectation of the sample coverage $C(p)$
has the minimum value when all $p_i$ are equal to $1 / L$, provided $N + 1 < L$.
\end{prop}

\begin{proof}
Our proof contains two parts. First we show that there exists only one extreme point for $\E(C(p))$, which has
$p_1 = p_2 = \cdots = p_L = 1 / L$. Second we prove that the minimum value can not be archived at the boundary.
Therefore the extreme point must correspond to the minimum value of $\E(C(p))$.

Recall the expectation of sample coverage can be expressed as
\[
\E(C(p)) = \sum_{i=1}^L p_i \left(1- (1 - p_i)^N \right) = 1 - \sum_{i=1}^L p_i(1 - p_i)^N,
\]
where $0 \leq p_i \leq 1$ for $i = 1, 2, \ldots, L$ and
\[
\sum_{i = 1}^L p_i = 1.
\]
To obtain the extreme values of $\E(C(p))$, we apply the method of Lagrange multipliers
\[
f_1(p_1, p_2, \ldots, p_L; \lambda) = 1 - \sum_{i=1}^L p_i(1 - p_i)^N + \lambda \left(\sum_{i=1}^L p_i -1\right).
\]
Set the partial derivative of $f_1$ for the variable $p_i$ to be zero:
\[
\frac{\partial f_1}{\partial p_i} = -(1 - p_i)^N + Np_i(1 - p_i)^{N-1} + \lambda = 0
\]
for $i = 1, 2, \ldots, L$.
For the convenience of analysis, we write the above equation as
\begin{equation}\label{eqn:samplecoverage}
(1 - p_i)^N - Np_i(1 - p_i)^{N-1} = \lambda.
\end{equation}
Now we analyze the roots of equation \eqref{eqn:samplecoverage}. Let $g_1(x)$ be the function
\[
g_1(x) = (1 - x)^N - Nx(1 - x)^{N-1}.
\]
The derivative of $g_1(x)$ is
\begin{align*}
\frac{\mathit{d}}{\mathit{dx}}g_1(x) &= -N(1 - x)^{N-1} - N(1 - x)^{N-1} + N(N-1)x(1-x)^{N-2} \\
&= N(1 - x)^{N-2}((N + 1)x - 2).
\end{align*}
From the derivative, the function $g_1(x)$ is decreasing for $x \in [0, 2 / (N+1)]$. It then becomes increasing
for $x \in [ 2 / (N+1), 1]$.
At endpoints, the function $g_1(x)$ has $g_1(0) > 0$, $g_1(1) = 0$ and $g_1(2 / (N+1)) < 0$.
Therefore equation \eqref{eqn:samplecoverage} has at most one root if $\lambda > 0$ and at most two roots
if $\lambda \leq 0$. We consider these two conditions separately. For the case where $\lambda > 0$,
since equation \eqref{eqn:samplecoverage} has at most one root, all $p_i$ must be equal to that solution.
Note that the sum of $p_i$ is $1$. We immediately obtain
\begin{equation}\label{eqn:sc_solution}
p_1 = p_2 = \cdots = p_L = 1 / L
\end{equation} and
\begin{equation}\label{eqn:sc_lambda}
\lambda =  \left(1 - \frac{1}{L}\right)^N - \frac{N}{L}\left(1 - \frac{1}{L}\right)^{N-1},
\end{equation}
if the extreme point exists. To confirm this is actually a valid solution, we need to show $\lambda > 0$.
Using the condition $(N+1) < L$, we then have
\begin{align*}
\lambda &=  \left(1 - \frac{1}{L}\right)^N - \frac{N}{L}\left(1 - \frac{1}{L}\right)^{N-1} \\
&=  \left(1 - \frac{1}{L}\right)^{N-1}\left( 1 - \frac{N+1}{L} \right) > 0.
\end{align*}
Thus when $\lambda > 0$, the expectation of the sample coverage $C(p)$ has only one extreme point.
Now we consider the case $\lambda \leq 0$ where equation \eqref{eqn:samplecoverage} could have two roots.
Note that $g_1(1 / (N + 1)) = 0$. Thus $g_1(x) > 0$ when $x$ is between $0$ and $1 / (N + 1)$. If $p_i$ is the solution of
equation \eqref{eqn:samplecoverage}, it must satisfy
\[
p_i \geq \frac{1}{N+1}.
\]
Suppose there exists an extreme point $(p_1, p_2, \ldots, p_L)$ when $\lambda \leq 0$. One should have
\[
\sum_{i=1}^L p_i \geq \frac{L}{N+1}.
\]
By condition $N + 1 < L$, we obtain
\[
\sum_{i=1}^L p_i > 1,
\]
which contradicts the fact that the sum should be $1$. Therefore when $\lambda \leq 0$, no extreme points exist. Combining
two cases $\lambda > 0$ and $\lambda \leq 0$, we conclude that \eqref{eqn:sc_solution} is the only extreme point
for $\E(C(p))$. At this point the value of $\E(C(p))$ is
\begin{equation}\label{eqn:sc_minimum}
1 - \left(1 - \frac{1}{L}\right)^N.
\end{equation}

Next we use mathematical induction to show that the minimum value of $\E(C(p))$ can not be archived at
the boundary for $L \geq 2$. As a result, the extreme value \eqref{eqn:sc_minimum} is actually the minimum value of $\E(C(p))$.
The statement is trivial when $L = 2$. Assume that when $L = k$, the minimum
value of $\E(C(p))$ does not lie in the boundary. So $\E(C(p))$ must archive the minimum
at an extreme point. Since there exists only one extreme point,
when $L = k$, the expectation of sample coverage $C(p)$ obtains the minimum value
\[
1 - \left(1 - \frac{1}{k}\right)^N,
\]
provided $p_i = 1 / k$ for $i = 1, 2, \ldots, k$. Now consider the case $L = k + 1$.
Suppose $\E(C(p))$ has the minimum value at the boundary. Without loss of generality, let $p_{k+1} = 0$.
Then the expectation of sample coverage $C(p)$ can be expressed as
\[
\E(C(p)) = 1 - \sum_{i=1}^{k} p_i(1 - p_i)^N,
\]
which becomes the case $L = k$.
As we have discussed, when $L = k$ the minimum value is
\[
1 - \left(1 - \frac{1}{k}\right)^N.
\]
However, if $p_i = 1 / (k+1)$ for $i = 1, 2, \ldots, k+1$, the expectation of
sample coverage $\E(C(p))$ becomes
\[
1 - \left(1 - \frac{1}{k+1}\right)^N,
\]
which is smaller than
\[
1 - \left(1 - \frac{1}{k}\right)^N.
\]
Therefore the minimum value of the expectation of sample coverage can
not lie at the boundary for $L = k + 1$. By mathematical induction,
this statement is true for all $L \geq 2$.  We conclude that the
extreme value \eqref{eqn:sc_minimum} must be the minimum value of
$\E(C(p))$.
\end{proof}


\section{DNA sequencing data preprocessing}
\label{sec:DNA}

Four publicly available single-cell DNA sequencing datasets were
downloaded from National Center for Biotechnology Information
(NCBI). The Sequencing Read Archive accession numbers are SRX202787,
SRX205367, SRX204160 and SRX151616
\citep{zong2012genome,lu2012probing,wang2012genome}.  
For each dataset, we subsampled 5M single-end reads as an initial sample. 
The sample was mapped to
human genome reference GRCh38/hg38 by BWA using default parameters
\cite{li2009fast}. We used Picard tools \cite{picard} to sort mapped
reads and Samtools \cite{li2009sequence} to obtain the number of
nucleotides covered by $j$ sequenced reads, for $j = 1, 2, \ldots$
These formed the counts $N_j$, the number of genomic sites $N_j$ covered by
exactly $j$ reads.


\newpage

\begin{figure}[h]
\centerline{\includegraphics{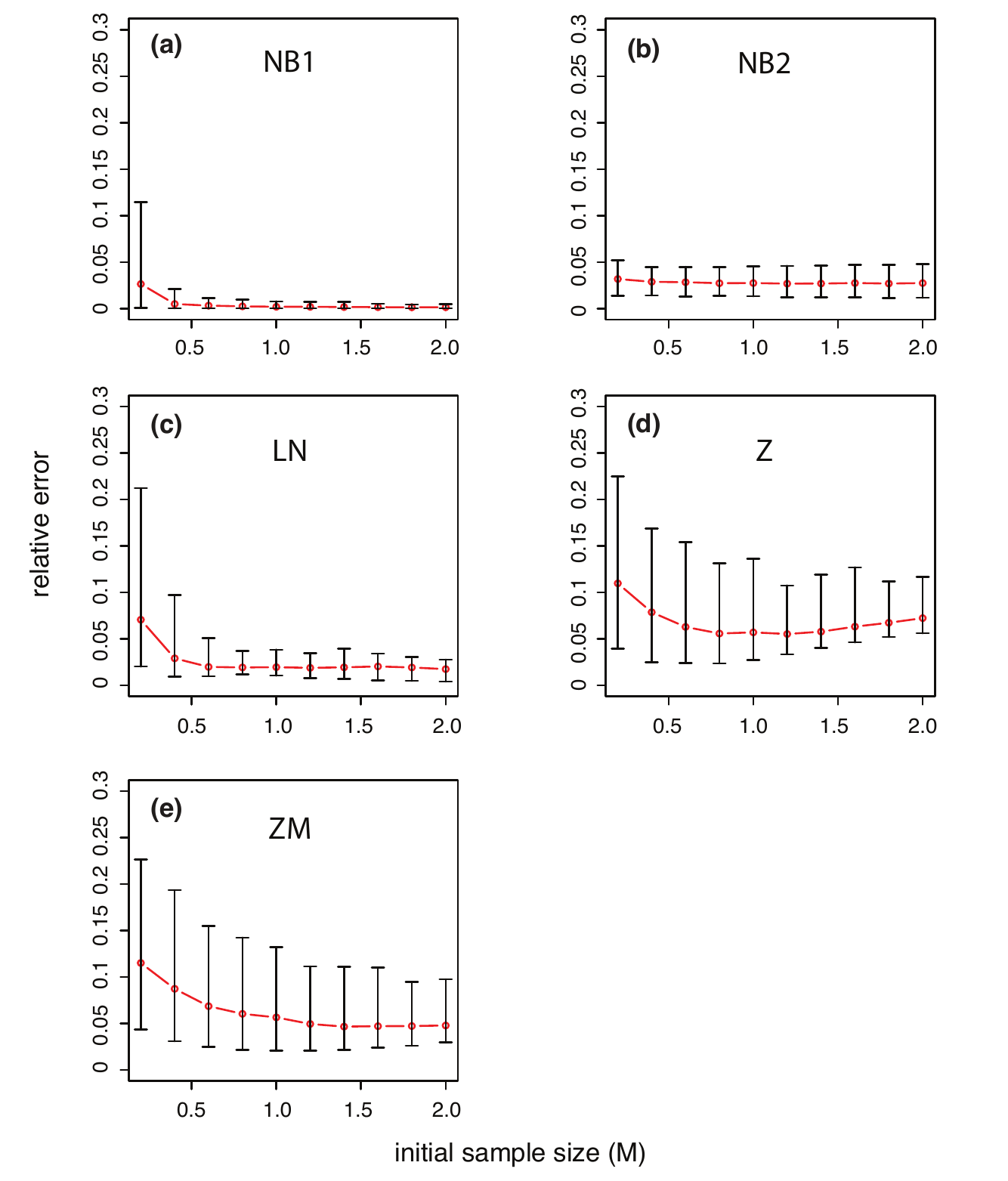}}
\caption{Relative error as a function of the initial sample size. The relative
  error of the estimator \dengnohat{} as the initial sample varies
  from $0.2L$ to $2L$. The $x$-axis is the initial sample size. The
  red line is the mean of relative error over 1000 replicates. The
  error bars show the 95\% confidence interval of relative errors.}
\label{fig:ds_sz}
\end{figure}


\begin{figure}[h]
\centerline{\includegraphics{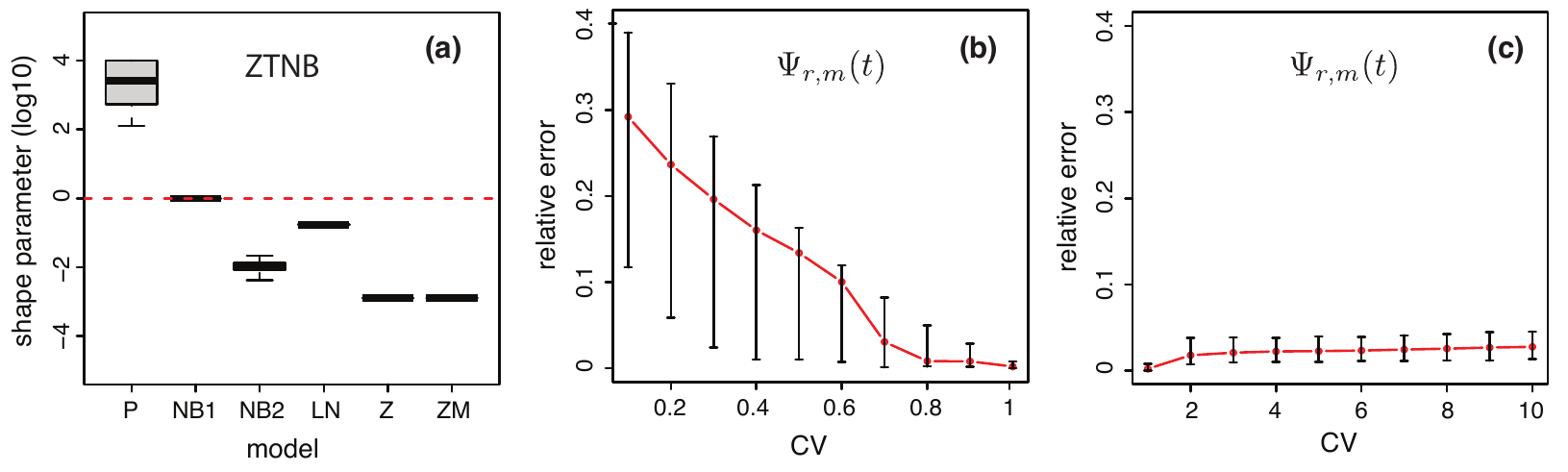}}
\caption{Degrees of the heterogeneity and effects on relative
  errors. (a) Estimated shape parameter $k$.  The estimated shape
  parameter $k$ was obtained by fitting a zero-truncated negative
  binomial (ZTNB) distribution to the data. Results are obtained by
  simulation 1000 samples for each model.  The $x$-axis represents
  model id. The $y$-axis is the logarithm of estimated shape
  parameters, with base $10$. The red dash line is the cutoff we use
  to distinguish heterogeneous populations from homogeneous
  populations.  (b, c) The mean relative error under NB models. We
  vary the shape parameter in NB models to simulate populations with
  different values of \cv{}. The red line is the mean relative error
  of the estimator \dengnohat{} as a function of the value of
  \cv{}. The error bar is the $95\%$ confidence interval of relative
  errors. (b) The range of \cv{} is from $0.1$ to $1$. (c) The range
  of \cv{} is from $1$ to $10$.}
\label{fig:cutoff_and_err}
\end{figure}


\begin{figure}[h]
\centerline{\includegraphics{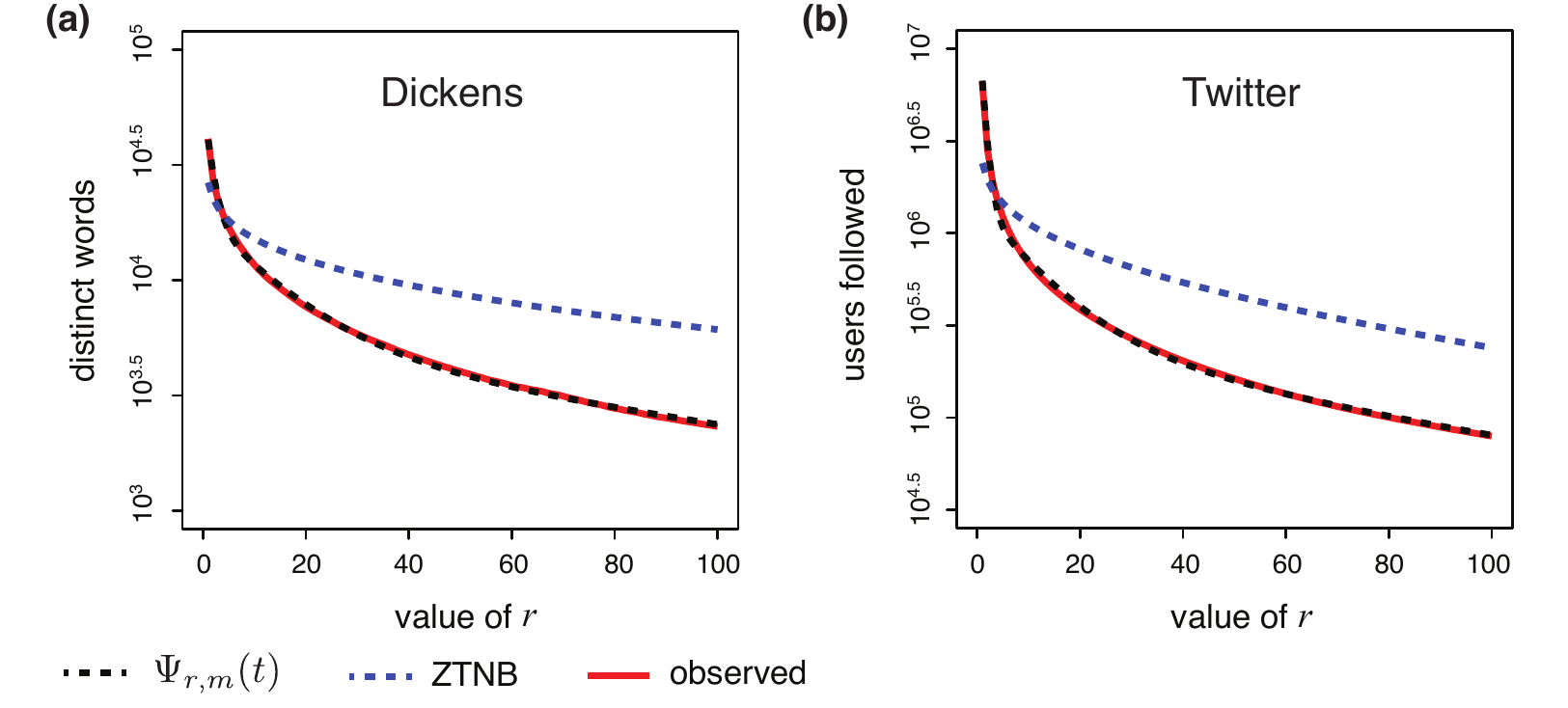}}
\caption{Observations versus estimation.  Species observed at least
  $r$ times along with estimated values for $r$ up to 100. (a)
  Unique words represented at least $r$ times along with estimated
  values for $t = 9$. (b) Users with at least $r$ followers
  along with estimated values for $t = 17$. The estimators
  \dengnohat{} and ZTNB are compared.}\label{fig:larger}
\end{figure}

\begin{figure}[h]
  \centerline{\includegraphics[scale=0.9]{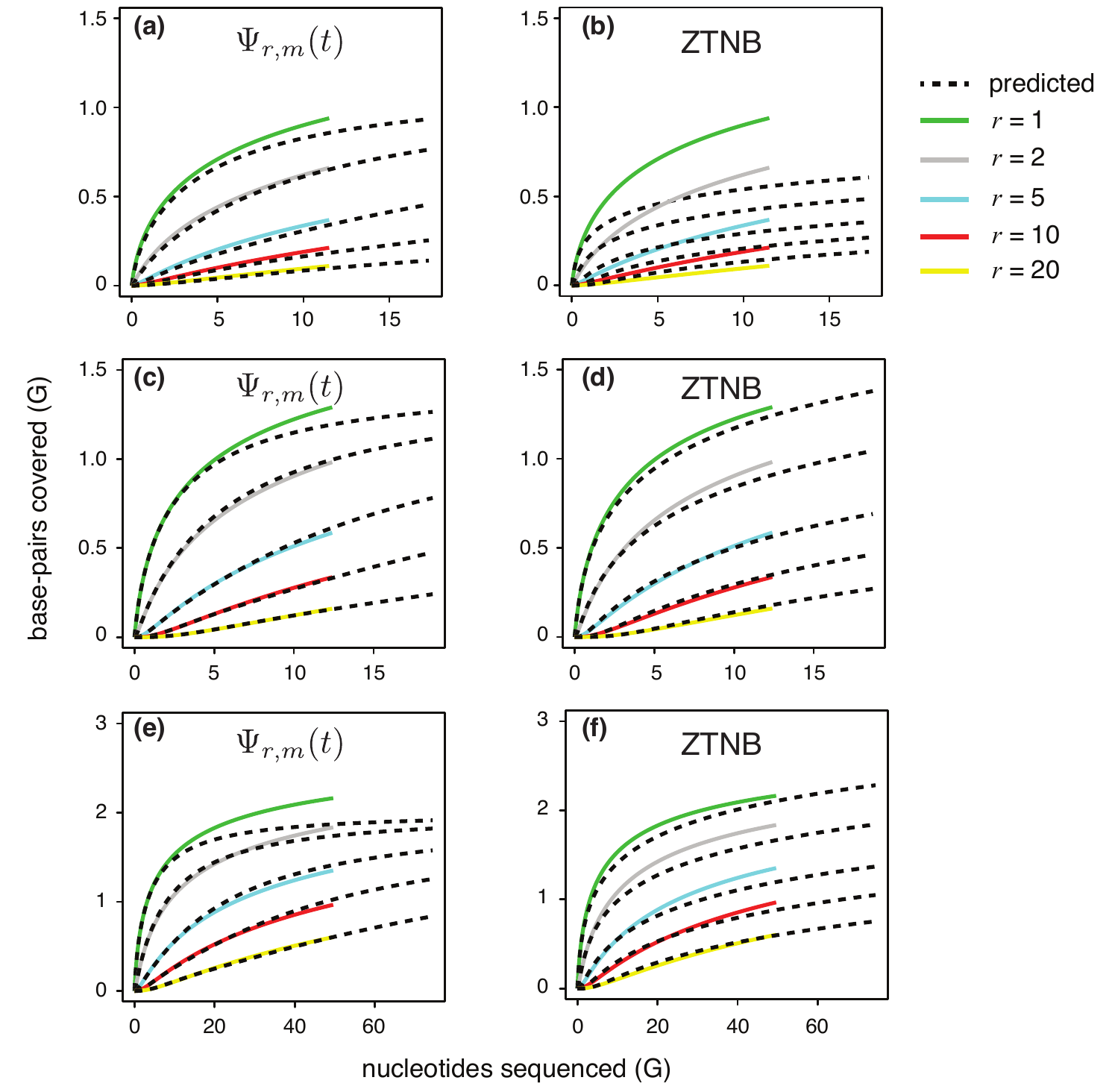}}
  \caption{Accuracy of estimates on DNA sequencing data.
    The number of base pairs in the genome represented $r$ or more times
    in the sequencing data set. Estimates were made using an initial
    sample of 5M reads (500M nucleotides).  Estimates for data set
    SRX151616 made using \dengnohat{} (a) and ZTNB (b).
    Estimates for data set SRR618566 made using \dengnohat{} (c) and
    ZTNB (d).  Estimates for data set SRX204160 made using
    \dengnohat{} (e) and ZTNB (f).}
\label{fig:alldna}
\end{figure}

\clearpage

\begin{table}
%
%
\caption{The coefficient of variation (CV) for datasets used in the main text.} \label{tab:cv}
\begin{tabular*}{\columnwidth}{@{\extracolsep{\fill}}*{7}l} 
\\ \hline
Shakespeare & Dickens & Twitter & SRX202787 & SRX205367 & SRX204160 & SRX151616\\
84.82 & 83.89 & 85.60 & 1.31 & 29.44 & 5.00 & 3.03 \\
\hline
\end{tabular*}

  \vspace*{1em}

\caption{Shakespeare's word use frequencies. Entry $r$ is $S_r$,
  the number of words that appear at least $r$ times in
  Shakespeare's known work.}\label{tab:shakespeareS}
\begin{tabular*}{\columnwidth}{@{\extracolsep{\fill}}l*{10}r}
\\ \hline
$r$ & 1 & 2 & 3 & 4 & 5 & 6 & 7 & 8 & 9 & 10 \\ \hline
+0 & 31534 & 17158 & 12815 & 10523 & 9060 & 8017 & 7180 & 6542 & 6023 & 5593 \\
+10 & 5229 & 4924 & 4665 & 4423 & 4200 & 4013 & 3832 & 3653 & 3523 & 3396 \\
\hline
\end{tabular*}
 
\end{table}

\begin{table}
\caption{Accuracy in predicting Dickens' frequency of word use.
  Estimates from \dengnohat{} for the number of unique words used at
  least $r$ times in a text along with standard errors (SD) and
  relative errors.  Estimates are shown for extrapolating to a text
  that is $5\times$ and $9\times$ larger than the initial $300$k word
  sample.  The SD is estimated by 100 bootstrap samples.
  The relative error is the absolute
  deviation between the predicted value and the observed value,
  dividing by the observed value.}\label{tab:dickens}
\begin{tabular*}{\columnwidth}{@{\extracolsep{\fill}}l*{6}r}
\\ \hline
 & \multicolumn{3}{c}{$5\times$}  & \multicolumn{3}{c}{$9\times$}\\ \cline{2-4} \cline{5-7}
 $r$ & predicted & SD & relative error & predicted & SD & relative error \\
\hline
1 & 33,196 & 1,292 & 0.008 & 40,188 & 2,438 & 0.008\\
2 & 21,525 & 1,425 & 0.022 & 28,323 & 1,922 & 0.070\\
5 & 12,364 & 677 & 0.013 & 16,008 & 1,025 & 0.038\\
10 & 8,337 & 508 & 0.020 & 11,462 & 673 & 0.007\\
20 & 4,956 & 432 & 0.018 & 7,626 & 500 & 0.020\\
\hline
\end{tabular*}
\vspace*{1em}

 \caption{Accuracy in predicting number of Twitter users with multiple followers.
  Estimates from \dengnohat{} for the number of users with $r$ or more followers
  along with standard errors (SD) and relative errors.
  Estimates are shown for extrapolating to a social network that is $5\times$ and $17\times$
  larger than the initial 5M following relationship sample.
  The SD is estimated by 100 bootstrap samples. The relative error is the absolute deviation between the predicted
  value and the observed value, dividing by the observed value.}\label{tab:twitter}
 \begin{tabular*}{\columnwidth}{@{\extracolsep{\fill}}l*{6}r}
 \\ \hline
& \multicolumn{3}{c}{$5\times$} & \multicolumn{3}{c}{$17\times$}\\ \cline{2-4} \cline{5-7}
$r$  & predicted & SD & relative error &  predicted & SD & relative error  \\ \hline
1 & 3,238,997 & 46,264 & 0.010 & 6,716,180 & 407,807 & 0.016\\
2 & 1,199,396 & 60,906 & 0.033 & 3,104,208 & 211,245 & 0.092\\
5 & 506,117 & 15,444 & 0.033 & 1,076,111 & 86,866 & 0.119\\
10 & 249,475 & 5,874 & 0.009 & 708,016 & 56,189 & 0.032\\
20 & 122,519 & 2,450 & 0.001 & 398,707 & 22,807 & 0.036\\
\hline
\end{tabular*}
\vspace*{1em}

  \caption{Accuracy in predicting number of base pairs covered by $r$
    or more reads. Estimates from \dengnohat{} for the number of base
    pairs covered by $r$ or more reads along with standard errors (SD)
    and relative errors.  Estimates are shown for extrapolating to a
    DNA sequencing experiment that is $5\times$ and $100\times$ larger
    than the initial 5M reads.  The SD is estimated by 100 bootstrap samples. 
    The relative error is the absolute deviation between the predicted value and the
    observed value, dividing by the observed value.}\label{tab:DNA}
 \begin{tabular*}{\columnwidth}{@{\extracolsep{\fill}}l*{6}r}
& \multicolumn{3}{c}{$5\times$} & \multicolumn{3}{c}{$100\times$}\\ \cline{2-4} \cline{5-7}
$r$  & predicted & SD & relative error &  predicted & SD & relative error  \\ \hline
1 & 1,129,927,035 & 329,156 & 0.010 & 2,652,400,412 & 17,217,958 & 0.022\\
2 & 479,320,554 & 638,870 & 0.017 & 2,446,839,012 & 11,662,586 & 0.015\\
5 & 57,629,462 & 174,216 & 0.041 & 1,930,375,358 & 1,472,315 & 0.001\\
10 & 4,345,056 & 67,038 & 0.110 & 1,323,241,725 & 5,343,486 & 0.003\\
20 & 717,334 & 70,268 & 0.226 & 667,534,255 & 3,020,688 & 0.021\\
\hline
\end{tabular*}
 \end{table}

\clearpage
\newpage

\bibliographystyle{chicago}
\bibliography{biblio}